\documentclass[journal]{IEEEtran}

\usepackage{amsmath,amssymb,amsfonts}
\usepackage{mathtools}
\usepackage{bm}
\usepackage{graphicx}
\usepackage{cite}
\usepackage{booktabs}
\usepackage{multirow}
\usepackage{array}
\usepackage{url}
\usepackage{hyperref}
\usepackage{placeins}

\usepackage{titlesec}

\titlespacing*{\section}
{0pt}{2pt}{2pt}

\titlespacing*{\subsection}
{0pt}{2pt}{1pt}

\titlespacing*{\subsubsection}
{0pt}{1pt}{1pt}

\usepackage{graphicx}
\usepackage{float}
\usepackage{subcaption}

\usepackage{float}

\usepackage{xcolor}
\usepackage{amsmath,amssymb,amsfonts}
\usepackage{amsmath,amssymb,amsfonts}
\usepackage{mathtools}
\usepackage{bm}
\usepackage{graphicx}
\usepackage{cite}
\usepackage{booktabs}
\usepackage{multirow}
\usepackage{array}
\usepackage{xcolor}

\usepackage{hyperref}

\hypersetup{
    colorlinks=true,
    linkcolor=blue,
    citecolor=blue,
    urlcolor=blue
}
\usepackage{tikz}
\usetikzlibrary{arrows.meta}

\makeatletter
\def\@IEEEsectpunct{\quad}
\makeatother

\usepackage{amsthm}

\newtheoremstyle{ieeestyle}
{0.5em}        
{0.5em}        
{\normalfont}  
{}             
{\bfseries}    
{:}            
{0.5em}        
{}             

\theoremstyle{ieeestyle}

\newtheorem{lemma}{Lemma}
\newtheorem{assumption}{Assumption}
\newtheorem{definition}{Definition}
\newtheorem{problem}{Problem}
\newtheorem{remark}{Remark}

\title{Distributed Adaptive Neural Interval Observers
for Unknown Nonlinear Systems}

\begin{document}

\author{
Tien Dat Vu,
My Nguyen Bach, Phuoc Vinh Nguyen
and Minh Doan%
\thanks{
T. D. Vu, P. V. Nguyen and M. Doan are with the Faculty of Mechanical Engineering,
Ho Chi Minh City University of Technology (HCMUT),
Vietnam National University Ho Chi Minh City (VNU-HCM),
Ho Chi Minh City, Vietnam
(e-mails: dat.vuv@hcmut.edu.vn; minh.doan@hcmut.edu.vn).
}%
\thanks{
M. N. Bach is with the Department of Mathematics and Statistics,
University of New Mexico,
Albuquerque, NM, USA
(e-mail: ngocmynguyenbach@gmail.com).
}%
\thanks{
Corresponding author: Minh Doan
(e-mail: minh.doan@hcmut.edu.vn).
}
}

\maketitle



\begin{abstract}
This paper develops a distributed adaptive neural interval observer for
unknown nonlinear systems with locally incomplete measurements. Each
sensor node constructs lower and upper state estimates using its local
output and neighboring observer information, while unknown nonlinear
dynamics are approximated by adaptive neural models. A 
distributed adaptation mechanism guarantees bounded estimation and
weight errors, whereas a cooperative realization preserves the
componentwise interval property. To enhance neural-weight convergence
without requiring persistent excitation, a finite experience-replay
integral concurrent-learning mechanism is incorporated into the
adaptation law. For non-Metzler error dynamics, a
Sylvester-based coordinate transformation is introduced to recover a
Hurwitz--Metzler distributed realization. The theoretical developments
are validated through a nonlinear distributed estimation example.
\end{abstract}


\section{Introduction}
\label{sec:introduction}

Reliable state estimation becomes difficult when a nonlinear plant is
observed through spatially distributed sensors, each having access only
to partial output information. A conventional distributed observer may
recover the state through local innovations and inter-node information
exchange, but it generally provides a point estimate rather than a
guaranteed state enclosure. Conversely, interval observers construct
\(\underline x(t)\leq x(t)\leq\bar x(t)\) under bounded uncertainty,
with the enclosure typically obtained from cooperative error dynamics
\cite{Gouze2000,Efimov2013,ZhangShen2023}. The problem considered here
is therefore to preserve such an interval simultaneously at every
sensor node when the nonlinear dynamics are unknown and no individual
sensor is required to be detectable.

Distributed state estimation has been extensively studied for linear
systems with partial sensing and network communication
\cite{Mitra2018,Han2019}. More recently, distributed interval observers
have been developed for uncertain LTI and descriptor systems under
joint or collective sensing conditions
\cite{Wang2025Distributed,Wang2026Descriptor}. In parallel, neural
networks have been incorporated into interval observers to approximate
unknown nonlinear functions \cite{Yi2022}. These developments,
however, leave a gap between distributed interval estimation and
adaptive learning of unknown nonlinear dynamics: existing distributed
interval constructions remain essentially model-based, whereas neural
interval observers are primarily centralized and their adaptation is
not derived from the distributed Lyapunov structure. Moreover,
parameter convergence in conventional adaptive schemes typically
requires persistent excitation, which can be restrictive in practical
online estimation. Concurrent-learning methods alleviate this
requirement by repeatedly exploiting sufficiently informative recorded
data, while integral concurrent learning further replaces the
persistent-excitation condition by an online-verifiable finite-data
richness condition \cite{Chowdhary2010,Parikh2019}.

This paper addresses this gap by constructing a distributed adaptive
neural interval observer under collective detectability. The main
contributions are summarized as follows. \emph{First}, upper and lower
neural observers are coupled over the sensing graph so that each node
uses only \(y_i\) and neighboring estimates, while local detectability
of every \((A,C_i)\) is not required. \emph{Second}, the neural
weight-update laws are derived directly from a network Lyapunov
function, yielding uniform ultimate boundedness of both estimation and
weight errors without introducing an independent training loss; a
finite experience-replay integral concurrent-learning mechanism is
further incorporated to strengthen neural-weight convergence without
requiring persistent excitation of the online regressors.
\emph{Third}, interval preservation is treated separately from
stability through a cooperative network realization; when the direct
distributed error dynamics are non-Metzler, a Sylvester-based
coordinate transformation is used to recover a Hurwitz--Metzler
realization and hence the componentwise enclosure
\(\underline x_i(t)\leq x(t)\leq\bar x_i(t)\) for all nodes.

\hypersetup{
    colorlinks=true,
    linkcolor=blue,
    citecolor=blue,
    urlcolor=blue
}


\section{Problem Formulation and Preliminaries}
\label{sec:problem}

\subsection{Graph-Theoretic Preliminaries}
\label{subsec:graph_preliminaries}

The communication among the \(N\) sensing nodes is described by a
weighted directed graph
\(\mathcal G=(\mathcal V,\mathcal E,\mathcal A)\), where
\(\mathcal V:=\{1,\ldots,N\}\) is the node set,
\(\mathcal E\subseteq\mathcal V\times\mathcal V\) is the directed edge
set, and
\(\mathcal A=[a_{ij}]\in\mathbb R^{N\times N}\) is the weighted
adjacency matrix. An edge \((j,i)\in\mathcal E\) means that node \(i\)
receives information from node \(j\), in which case \(a_{ij}>0\);
otherwise \(a_{ij}=0\), and \(a_{ii}=0\). The in-neighbor set of node
\(i\) is
\(\mathcal N_i:=\{j\in\mathcal V:(j,i)\in\mathcal E\}
=\{j\in\mathcal V:a_{ij}>0\}\), and its weighted in-degree is
\(d_i:=\sum_{j=1}^{N}a_{ij}\). The corresponding degree matrix is
\(\mathcal D:=\operatorname{diag}\{d_1,\ldots,d_N\}\), and the graph
Laplacian is defined as
\(\mathcal L:=\mathcal D-\mathcal A\). Hence,
\(\mathcal L_{ii}=d_i\) and
\(\mathcal L_{ij}=-a_{ij}\) for \(i\neq j\), so that
\(\mathcal L\mathbf 1_N=0\). Throughout the paper, information
available to node \(i\) consists of its local measurement together with
the signals received from nodes in \(\mathcal N_i\).

\subsection{Networked Nonlinear System}
\label{subsec:system}

Consider an unknown nonlinear plant monitored by \(N\) distributed
sensing nodes,
\begin{equation}
\dot{x}=Ax+Bf(x,u)+Ed,
\qquad
y_i=C_i x,\quad i\in\mathcal V,
\label{eq:distributed_plant}
\end{equation}
where \(x\in\mathbb R^n\), \(u\in\mathbb R^m\),
\(d\in\mathbb R^q\), \(y_i\in\mathbb R^{p_i}\), and
\(\mathcal V:=\{1,\ldots,N\}\). The matrices
\(A\in\mathbb R^{n\times n}\),
\(B\in\mathbb R^{n\times r}\),
\(E\in\mathbb R^{n\times q}\), and
\(C_i\in\mathbb R^{p_i\times n}\) are known, whereas
\(f:\mathbb R^n\times\mathbb R^m\rightarrow\mathbb R^r\) is an
unknown nonlinear mapping.


In observer-based control, the plant is typically operated under a
feedback policy. In the state-feedback case \(u=\pi(x)\), the unknown
nonlinearity becomes the composite state-dependent mapping
\(f_c(x):=f(x,\pi(x))\). Hence, instead of approximating a general
two-argument map \(f(x,u)\), the neural model may approximate
\(f_c(x)\) directly as
\(f_c(x)=W^{*\top}\phi(x)+\varepsilon(x)\).
Accordingly, the subsequent neural basis functions may be chosen as
functions of the state argument only.

All vector inequalities are understood componentwise. For any matrix
\(M=[m_{jk}]\), define
\(M^{+}:=[\max\{m_{jk},0\}]\) and
\(M^{-}:=[\max\{-m_{jk},0\}]\), such that
\(M=M^{+}-M^{-}\) with \(M^{+},M^{-}\geq0\).

\begin{lemma}
\label{lem:matrix_interval}
For any \(\underline z\leq z\leq\bar z\) and any compatible matrix
\(M\),
\(M^{+}\underline z-M^{-}\bar z
\leq Mz\leq
M^{+}\bar z-M^{-}\underline z\).
\end{lemma}

\begin{IEEEproof}
Since \(M=M^{+}-M^{-}\) and \(M^{+},M^{-}\geq0\), the result follows
directly from
\(M^{+}\underline z\leq M^{+}z\leq M^{+}\bar z\) and
\(M^{-}\underline z\leq M^{-}z\leq M^{-}\bar z\).
\end{IEEEproof}

\begin{assumption}
\label{ass:disturbance}
The disturbance is unknown but satisfies
\(\underline d(t)\leq d(t)\leq\bar d(t)\), where
\(\underline d,\bar d\in\mathbb R^q\) are known and locally bounded.
\end{assumption}

From Lemma~\ref{lem:matrix_interval}, define
\(\underline\psi_d:=E^{+}\underline d-E^{-}\bar d\) and
\(\bar\psi_d:=E^{+}\bar d-E^{-}\underline d\). Then
\(\underline\psi_d\leq Ed\leq\bar\psi_d\), so the actual value of
\(d(t)\) is not required online.

\begin{assumption}
\label{ass:compact_region}
The trajectories satisfy
\(x(t)\in\Omega_x\subset\mathbb R^n\) and
\(u(t)\in\Omega_u\subset\mathbb R^m\), where
\(\Omega_x\) and \(\Omega_u\) are compact. The mapping \(f(x,u)\) is
continuous on \(\Omega_x\times\Omega_u\) and locally Lipschitz in
\(x\), uniformly in \(u\) (In observer-based control, the observer is typically operated together
with a stabilizing controller, whose role is also to keep the closed-loop
state trajectory inside the prescribed compact set \(\Omega_x\)).
\end{assumption}

For node \(i\), let
\(\underline x_i,\bar x_i\in\mathbb R^n\) denote the lower and upper
state estimates. The unknown nonlinear term is represented through
two adaptive neural models.

\begin{assumption}
\label{ass:neural_interval}
For each \(i\in\mathcal V\), there exist bounded basis vectors
\(\underline\phi_i
=\underline\phi_i(\underline x_i,\bar x_i,u)\in\mathbb R^{N_i}\),
\(\bar\phi_i
=\bar\phi_i(\underline x_i,\bar x_i,u)\in\mathbb R^{N_i}\),
finite ideal weights
\(\underline W_i^{*},\bar W_i^{*}\in\mathbb R^{N_i\times r}\),
and \(\varepsilon_{fi}\in\mathbb R_{\geq0}^{r}\) such that
\begin{equation}
\underline W_i^{*\top}\underline\phi_i-\varepsilon_{fi}
\leq f(x,u)\leq
\bar W_i^{*\top}\bar\phi_i+\varepsilon_{fi},
\label{eq:distributed_nn_interval}
\end{equation}
whenever
\(\underline x_i\leq x\leq\bar x_i\) and
\((x,u)\in\Omega_x\times\Omega_u\). Moreover,
\(\|\underline\phi_i\|\leq\bar\phi_i^{\ell}\) and
\(\|\bar\phi_i\|\leq\bar\phi_i^{u}\) for some finite
\(\bar\phi_i^{\ell},\bar\phi_i^{u}>0\).
\end{assumption}

The ideal weights in
Assumption~\ref{ass:neural_interval} are used only for analysis and
are unavailable online. Their estimates
\(\underline W_i(t)\) and \(\bar W_i(t)\) will be updated through the
Lyapunov-derived adaptive laws developed in
Section~\ref{sec:observer}. Since \(B\) may be sign-indefinite, the
bounds in \eqref{eq:distributed_nn_interval} will later be propagated
through \(B\) using Lemma~\ref{lem:matrix_interval}.


\subsection{Distributed Sensing and Communication}
\label{subsec:network}


Node \(i\) measures only \(y_i=C_i x\); hence, the individual pair
\((A,C_i)\) is not required to be detectable. Define the aggregate
output matrix
\(C_{\mathcal V}:=\operatorname{col}\{C_1,\ldots,C_N\}\) and
\(y_{\mathcal V}:=\operatorname{col}\{y_1,\ldots,y_N\}
=C_{\mathcal V}x\). The aggregate quantity \(y_{\mathcal V}\) is
introduced only for analysis and is not available at any individual
node.

\begin{definition}[Collective detectability]
\label{def:collective_detectability}
The sensing network is collectively detectable if
\((A,C_{\mathcal V})\) is detectable, i.e., for every
\(\lambda\in\sigma(A)\) with
\(\operatorname{Re}(\lambda)\geq0\),
\(\operatorname{rank}
\operatorname{col}\{\lambda I_n-A,C_{\mathcal V}\}=n\).
\end{definition}

Thus, an unstable mode may be unobservable from some \(C_i\), but it
must be detectable from the sensing network as a whole.

\begin{assumption}
\label{ass:graph}
The graph \(\mathcal G\) is fixed and strongly connected, with
\(a_{ij}>0\) for every \((j,i)\in\mathcal E\).
\end{assumption}

\begin{assumption}
\label{ass:collective_detectability}
The pair \((A,C_{\mathcal V})\) is collectively detectable according
to Definition~\ref{def:collective_detectability}.
\end{assumption}

At node \(i\), the admissible online information is
\(\{y_i,u,\underline d,\bar d,
\underline x_i,\bar x_i,
\underline W_i,\bar W_i,
\underline x_j,\bar x_j:
j\in\mathcal N_i\}\).
Thus, no node requires the full output \(y_{\mathcal V}\) or the true
state \(x\).

\begin{assumption}
\label{ass:communication}
Neighboring information is transmitted continuously without delay or
packet loss. The communication topology is fixed during the interval
of operation.
\end{assumption}


Collective detectability and interval preservation are distinct but complementary requirements in the distributed interval observer design. The former ensures that the sensing network as a whole contains sufficient information to detect the relevant plant modes through the pair \((A,C_{\mathcal V})\), whereas the latter guarantees that the upper and lower estimates preserve the required componentwise ordering over the entire network. In particular, the cooperative error matrix in \eqref{eq:network_error_matrix} must be Metzler, i.e., all of its off-diagonal entries are nonnegative~\cite{Gouze2000,Efimov2013,ZhangShen2023}, in order for the distributed interval observer to preserve the interval enclosure.


\subsection{Problem Statement}
\label{subsec:problem_statement}

For each \(i\in\mathcal V\), define the upper and lower estimation
errors as
\(\bar e_i:=\bar x_i-x\) and
\(\underline e_i:=x-\underline x_i\). Hence,
\(\underline x_i\leq x\leq\bar x_i\) is equivalent to
\(\underline e_i,\bar e_i\geq0\). At the network level, let
\(\bar e:=\operatorname{col}\{\bar e_1,\ldots,\bar e_N\}\) and
\(\underline e
:=\operatorname{col}\{\underline e_1,\ldots,\underline e_N\}\).

Define the neural-weight errors by
\(\tilde{\bar W}_i:=\bar W_i-\bar W_i^{*}\) and
\(\tilde{\underline W}_i
:=\underline W_i-\underline W_i^{*}\).
The observer is initialized such that
\(\underline x_i(0)\leq x(0)\leq\bar x_i(0)\), or equivalently
\(\underline e_i(0),\bar e_i(0)\geq0\), for all
\(i\in\mathcal V\).

\begin{problem}
\label{prob:distributed_interval_observer}
For system~\eqref{eq:distributed_plant} under
Assumptions~\ref{ass:disturbance}--\ref{ass:communication}, construct
at each \(i\in\mathcal V\) distributed upper and lower adaptive neural
observers, together with online laws for
\(\underline W_i\) and \(\bar W_i\), such that all estimation and
weight errors are uniformly ultimately bounded and
\begin{equation}
\underline x_i(t)\leq x(t)\leq\bar x_i(t),
\qquad
\forall i\in\mathcal V,\quad t\geq0.
\label{eq:distributed_interval_objective}
\end{equation}
The design must use only the admissible local and neighboring
information.
\end{problem}

Problem~\ref{prob:distributed_interval_observer} therefore contains
two independent requirements: bounded adaptive estimation dynamics
and preservation of the order
\(\underline e,\bar e\geq0\). The first will be established by a
network Lyapunov analysis, whereas the second follows from a
cooperative realization of the distributed error dynamics (See Section~\ref{sec:observer}, Eq.~\ref{eq:network_error_matrix} and Section~\ref{subsec:cooperative_realization}).

\begin{remark}
\label{rem:stability_vs_interval}
Hurwitz stability alone does not imply interval preservation. A
non-Metzler error matrix may yield bounded or convergent estimation
errors while allowing some components of
\(\underline e\) or \(\bar e\) to become negative. Consequently,
stability and cooperativity must be established separately.
\end{remark}

\begin{remark}
\label{rem:distributed_nature}
Local detectability of every \((A,C_i)\) is not required. The design
instead exploits collective detectability together with graph
connectivity, so that information unavailable from \(y_i\) can be
supplied through neighboring observer interactions.
\end{remark}


\section{Distributed Adaptive Neural Interval Observer}
\label{sec:observer}


\subsection{Observer Architecture}
\label{subsec:observer_architecture}


As discussed in Remark~\ref{rem:stability_vs_interval}, the direct realization in the original
coordinates does not necessarily possess the Metzler structure required
for interval preservation. Therefore, a suitable coordinate
transformation is introduced to construct a Hurwitz--Metzler realization
whenever such a structure is not directly available. For each node
\(i\in\mathcal V\), introduce a nonsingular coordinate transformation
\(T_i\in\mathbb R^{n\times n}\) and define \(z_i:=T_i x\). Let
\(A_i:=T_iAT_i^{-1}\),
\(B_i:=T_iB\),
\(E_i:=T_iE\),
\(C_{zi}:=C_iT_i^{-1}\), and
\(R_{ij}:=T_iT_j^{-1}\).
Hence, the plant expressed in the coordinates of node \(i\) satisfies
\(\dot z_i=A_i z_i+B_i f(x,u)+E_i d\) and
\(y_i=C_{zi}z_i\). If the direct realization in the original coordinates
already satisfies the required Metzler and Hurwitz properties, no
coordinate transformation is necessary and one simply sets
\(T_i=I_n\). Otherwise, \(T_i\) is selected to recover the required
Hurwitz--Metzler realization, as developed in
Subsection~\ref{subsec:cooperative_realization}.

Let the online upper and lower neural models be
\(\bar f_i:=\bar W_i^\top\bar\phi_i+\bar\delta_i\) and
\(\underline f_i:=\underline W_i^\top\underline\phi_i-\underline\delta_i\),
where
\(\bar\delta_i,\underline\delta_i\in\mathbb R_{\geq0}^{r}\) satisfy
\(\bar\delta_i\geq\varepsilon_{fi}\) and
\(\underline\delta_i\geq\varepsilon_{fi}\).
Using the ideal weights of
Assumption~\ref{ass:neural_interval}, define
\(\bar\Delta_{fi}:=
\bar W_i^{*\top}\bar\phi_i+\bar\delta_i-f(x,u)\) and
\(\underline\Delta_{fi}:=
f(x,u)-\underline W_i^{*\top}\underline\phi_i+\underline\delta_i\).
Then
\(\bar\Delta_{fi}\geq0\) and
\(\underline\Delta_{fi}\geq0\).

Since \(B_i\) may be sign-indefinite, define
\(\bar\psi_i:=B_i^{+}\bar f_i-B_i^{-}\underline f_i\) and
\(\underline\psi_i:=B_i^{+}\underline f_i-B_i^{-}\bar f_i\).
Similarly, using Assumption~\ref{ass:disturbance}, define
\(\bar q_i:=E_i^{+}\bar d-E_i^{-}\underline d\) and
\(\underline q_i:=E_i^{+}\underline d-E_i^{-}\bar d\).
Lemma~\ref{lem:matrix_interval} gives
\(\underline q_i\leq E_i d\leq\bar q_i\).

The proposed distributed upper and lower observers are
\begin{equation}
\begin{aligned}
\dot{\bar z}_i
={}&
A_i\bar z_i+\bar\psi_i
+K_i(y_i-C_{zi}\bar z_i)+\bar q_i
\\
&+\kappa\sum_{j\in\mathcal N_i}
a_{ij}\left(R_{ij}\bar z_j-\bar z_i\right),
\\[1mm]
\dot{\underline z}_i
={}&
A_i\underline z_i+\underline\psi_i
+K_i(y_i-C_{zi}\underline z_i)+\underline q_i
\\
&+\kappa\sum_{j\in\mathcal N_i}
a_{ij}\left(R_{ij}\underline z_j-\underline z_i\right),
\end{aligned}
\label{eq:distributed_interval_observer}
\end{equation}
where \(K_i\in\mathbb R^{n\times p_i}\) is the output-injection gain
and \(\kappa>0\) is the distributed coupling gain.

The term \(R_{ij}\bar z_j\) is expressed in the coordinates of node
\(i\), since
\(R_{ij}z_j=T_iT_j^{-1}T_jx=T_ix=z_i\).
Thus, the coupling term vanishes when neighboring estimates coincide
with the true transformed state.

Define the transformed upper and lower errors by
\(\bar\epsilon_i:=\bar z_i-z_i\) and
\(\underline\epsilon_i:=z_i-\underline z_i\).
The measurable output residuals are
\(\bar r_i:=C_{zi}\bar z_i-y_i=C_{zi}\bar\epsilon_i\) and
\(\underline r_i:=y_i-C_{zi}\underline z_i
=C_{zi}\underline\epsilon_i\).

Using
\(\tilde{\bar W}_i:=\bar W_i-\bar W_i^*\) and
\(\tilde{\underline W}_i:=\underline W_i-\underline W_i^*\),
one has
\(\bar f_i-f
=\tilde{\bar W}_i^\top\bar\phi_i+\bar\Delta_{fi}\) and
\(f-\underline f_i
=-\tilde{\underline W}_i^\top\underline\phi_i
+\underline\Delta_{fi}\).

Define the nonnegative disturbance residuals \(\bar\Delta_{di}:=\bar q_i-E_i d\geq0\) and \(\underline\Delta_{di}:=E_i d-\underline q_i\geq0\), and let \(\bar\Delta_i:=B_i^{+}\bar\Delta_{fi}+B_i^{-}\underline\Delta_{fi}+\bar\Delta_{di}\) and \(\underline\Delta_i:=B_i^{+}\underline\Delta_{fi}+B_i^{-}\bar\Delta_{fi}+\underline\Delta_{di}\). Hence, \(\bar\Delta_i\geq0\) and \(\underline\Delta_i\geq0\).


Let \(M_i:=A_i-K_iC_{zi}\) and
\(d_i:=\sum_{j=1}^{N}a_{ij}\). Subtracting the transformed plant
dynamics from \eqref{eq:distributed_interval_observer} gives the upper
interval-error dynamics
\begin{equation}
\begin{aligned}
\dot{\bar\epsilon}_i
={}&
(M_i-\kappa d_iI_n)\bar\epsilon_i
+
\kappa\sum_{j\in\mathcal N_i}
a_{ij}R_{ij}\bar\epsilon_j
\\
&+
B_i^{+}\tilde{\bar W}_i^\top\bar\phi_i
-
B_i^{-}\tilde{\underline W}_i^\top\underline\phi_i
+
\bar\Delta_i ,
\end{aligned}
\label{eq:upper_local_interval_error}
\end{equation}
whereas the lower interval-error dynamics are
\begin{equation}
\begin{aligned}
\dot{\underline\epsilon}_i
={}&
(M_i-\kappa d_iI_n)\underline\epsilon_i
+
\kappa\sum_{j\in\mathcal N_i}
a_{ij}R_{ij}\underline\epsilon_j
\\
&+
B_i^{-}\tilde{\bar W}_i^\top\bar\phi_i
-
B_i^{+}\tilde{\underline W}_i^\top\underline\phi_i
+
\underline\Delta_i .
\end{aligned}
\label{eq:lower_local_interval_error}
\end{equation}

Define \(\bar\epsilon:=\operatorname{col}\{\bar\epsilon_1,\ldots,\bar\epsilon_N\}\), \(\underline\epsilon:=\operatorname{col}\{\underline\epsilon_1,\ldots,\underline\epsilon_N\}\), \(\bar\eta:=\operatorname{col}\{\tilde{\bar W}_1^\top\bar\phi_1,\ldots,\tilde{\bar W}_N^\top\bar\phi_N\}\), and \(\underline\eta:=\operatorname{col}\{\tilde{\underline W}_1^\top\underline\phi_1,\ldots,\tilde{\underline W}_N^\top\underline\phi_N\}\). Moreover, let \(\mathcal B_{+}:=\operatorname{blkdiag}\{B_1^{+},\ldots,B_N^{+}\}\), \(\mathcal B_{-}:=\operatorname{blkdiag}\{B_1^{-},\ldots,B_N^{-}\}\), \(\bar\Delta:=\operatorname{col}\{\bar\Delta_1,\ldots,\bar\Delta_N\}\), and \(\underline\Delta:=\operatorname{col}\{\underline\Delta_1,\ldots,\underline\Delta_N\}\).

Define \(\mathcal R_{\mathcal A}\in\mathbb R^{Nn\times Nn}\) by \([\mathcal R_{\mathcal A}]_{ii}:=0\) and \([\mathcal R_{\mathcal A}]_{ij}:=a_{ij}R_{ij}\) for \(i\neq j\), where \(R_{ij}:=T_iT_j^{-1}\). The network error matrix is therefore
\begin{equation}
\mathcal M_D
:=
\begin{bmatrix}
M_1-\kappa d_1I_n
&
\kappa a_{12}R_{12}
&
\cdots
&
\kappa a_{1N}R_{1N}
\\
\kappa a_{21}R_{21}
&
M_2-\kappa d_2I_n
&
\cdots
&
\kappa a_{2N}R_{2N}
\\
\vdots
&
\vdots
&
\ddots
&
\vdots
\\
\kappa a_{N1}R_{N1}
&
\kappa a_{N2}R_{N2}
&
\cdots
&
M_N-\kappa d_NI_n
\end{bmatrix},
\label{eq:network_error_matrix}
\end{equation}
Equivalently, its blocks satisfy \([\mathcal M_D]_{ii}=M_i-\kappa d_iI_n\) and \([\mathcal M_D]_{ij}=\kappa a_{ij}R_{ij}\) for \(i\neq j\), with \([\mathcal M_D]_{ij}=0\) whenever \(a_{ij}=0\). Thus, this compact representation is exactly equivalent to the full block matrix representation.

The complete error dynamics are therefore
\begin{equation}
\dot{\bar\epsilon}
=
\mathcal M_D\bar\epsilon
+\mathcal B_{+}\bar\eta
-\mathcal B_{-}\underline\eta
+\bar\Delta,
\qquad
\dot{\underline\epsilon}
=
\mathcal M_D\underline\epsilon
+\mathcal B_{-}\bar\eta
-\mathcal B_{+}\underline\eta
+\underline\Delta.
\label{eq:stacked_interval_error}
\end{equation}



\subsection{Cooperative Realization}
\label{subsec:cooperative_realization}

The interval property requires \(\mathcal M_D\) to be Metzler in addition to being Hurwitz. The distributed coupling itself is order preserving whenever \(R_{ij}\geq0\), since its off-diagonal blocks satisfy \(\kappa a_{ij}R_{ij}\geq0\). The remaining issue is therefore the local diagonal realization \(M_i\).

If \(T_i=I_n\) and gains \(K_i\) can be selected such that \(\mathcal M_D=\operatorname{blkdiag}\{A-K_iC_i\}_{i=1}^{N}-\kappa\mathcal L\otimes I_n\) is Hurwitz and Metzler, no coordinate transformation is required. Otherwise, the following Sylvester realization is used.


From
\eqref{eq:upper_local_interval_error}--\eqref{eq:lower_local_interval_error},
The target local error
dynamics are thus characterized by
\(A_i-K_iC_{zi}=M_i\). Using the previously defined transformed
matrices \(A_i=T_iAT_i^{-1}\) and \(C_{zi}=C_iT_i^{-1}\), this
condition becomes
\(T_iAT_i^{-1}-K_iC_iT_i^{-1}=M_i\). Right multiplication by \(T_i\)
gives \(T_iA-K_iC_i=M_iT_i\), and hence
\(T_iA-M_iT_i=K_iC_i\). Therefore, for each node \(i\), one prescribes
a Metzler matrix \(M_i\in\mathbb R^{n\times n}\) and an injection
matrix \(K_i\in\mathbb R^{n\times p_i}\), and obtains the coordinate
transformation \(T_i\) by solving the Sylvester relation
\(T_iA-M_iT_i=K_iC_i\).
\begin{lemma}
\label{lem:distributed_sylvester}
For fixed \(M_i\) and \(K_i\), if \(\sigma(A)\cap\sigma(M_i)=\varnothing\), then the Sylvester relation above admits a unique solution \(T_i\). If \(T_i\) is nonsingular, then with \(L_i:=T_i^{-1}K_i\), one has \(T_i(A-L_iC_i)T_i^{-1}=M_i\).
\end{lemma}


\begin{IEEEproof}
Consider the Sylvester equation \(T_iA-M_iT_i=K_iC_i\). Using the vectorization identity \(\operatorname{vec}(AXB)=(B^\top\otimes A)\operatorname{vec}(X)\), one has \(\operatorname{vec}(T_iA)=(A^\top\otimes I_n)\operatorname{vec}(T_i)\) and \(\operatorname{vec}(M_iT_i)=(I_n\otimes M_i)\operatorname{vec}(T_i)\). Hence, the Sylvester equation can be written as \(\left(A^\top\otimes I_n-I_n\otimes M_i\right)\operatorname{vec}(T_i)=\operatorname{vec}(K_iC_i)\). The eigenvalues of \(A^\top\otimes I_n-I_n\otimes M_i\) are \(\lambda_j(A)-\lambda_k(M_i)\), \(j,k=1,\ldots,n\). Therefore, the condition \(\sigma(A)\cap\sigma(M_i)=\varnothing\) implies \(\lambda_j(A)-\lambda_k(M_i)\neq0\) for all \(j,k\), and consequently \(0\notin\sigma(A^\top\otimes I_n-I_n\otimes M_i)\). Thus, \(A^\top\otimes I_n-I_n\otimes M_i\) is nonsingular, and the Sylvester equation admits the unique solution \(\operatorname{vec}(T_i)=\left(A^\top\otimes I_n-I_n\otimes M_i\right)^{-1}\operatorname{vec}(K_iC_i)\). If \(T_i\) is nonsingular, define \(L_i:=T_i^{-1}K_i\), so that \(K_i=T_iL_i\). From \(T_iA-M_iT_i=K_iC_i\), it follows that \(T_iA-K_iC_i=M_iT_i\). Substituting \(K_i=T_iL_i\) gives \(T_iA-T_iL_iC_i=M_iT_i\), or equivalently \(T_i(A-L_iC_i)=M_iT_i\). Right multiplication by \(T_i^{-1}\) then yields \(T_i(A-L_iC_i)T_i^{-1}=M_i\).
\end{IEEEproof}


\begin{assumption}
\label{ass:cooperative_realization}
There exist \(\kappa>0\), matrices \(M_i\), \(K_i\), and nonsingular Sylvester solutions \(T_i\) such that \(M_i\) is Metzler, \(\sigma(A)\cap\sigma(M_i)=\varnothing\), \(T_iT_j^{-1}\geq0\) for every \((j,i)\in\mathcal E\), and the distributed network matrix \(\mathcal M_D\) is Hurwitz.
\end{assumption}

\begin{remark}[Practical selection of the cooperative realization]
\label{rem:cooperative_design}
Lemma~\ref{lem:distributed_sylvester} and Assumption~\ref{ass:cooperative_realization} provide a practical offline guideline for selecting the realization. A convenient choice is to prescribe each \(M_i\) as a Hurwitz--Metzler matrix with stability margin \(\alpha_i:=-\max_{\mu\in\sigma(M_i)}\operatorname{Re}(\mu)>0\), while enforcing the spectral-separation margin \(\delta_i:=\min_{\lambda\in\sigma(A),\,\mu\in\sigma(M_i)}|\lambda-\mu|>0\). For a selected \(K_i\), Lemma~\ref{lem:distributed_sylvester} then guarantees a unique solution \(T_i\) of \(T_iA-M_iT_i=K_iC_i\). The resulting realization is admissible only if \(\sigma_{\min}(T_i)>0\), \(R_{ij}=T_iT_j^{-1}\geq0\) for every \((j,i)\in\mathcal E\), and the assembled distributed matrix satisfies the network stability margin \(\alpha_D:=-\max_{\lambda\in\sigma(\mathcal M_D)}\operatorname{Re}(\lambda)>0\), as required by Assumption~\ref{ass:cooperative_realization}. Hence, \(M_i\), \(K_i\), and \(\kappa\) may be tuned offline until \(\delta_i>0\), \(\sigma_{\min}(T_i)>0\), \(R_{ij}\geq0\), and \(\alpha_D>0\) hold simultaneously; choosing each \(M_i\) Hurwitz is a convenient sufficient design preference, although Assumption~\ref{ass:cooperative_realization} ultimately requires Hurwitz stability only of the complete matrix \(\mathcal M_D\).
\end{remark}

Notice that the individual matrix \(M_i\) is not required to be Hurwitz. Stability is imposed on the complete matrix \(\mathcal M_D\), so local sensing deficiencies may be compensated by the communication coupling. This is consistent with the collective detectability requirement of Assumption~\ref{ass:collective_detectability}.


Once the transformed interval satisfies \(\underline z_i\leq z_i\leq\bar z_i\), the physical-state interval is reconstructed according to \(\underline x_i=(T_i^{-1})^{+}\underline z_i-(T_i^{-1})^{-}\bar z_i\) and \(\bar x_i=(T_i^{-1})^{+}\bar z_i-(T_i^{-1})^{-}\underline z_i\).

Indeed, since \(x=T_i^{-1}z_i\), application of Lemma~\ref{lem:matrix_interval} directly gives \((T_i^{-1})^{+}\underline z_i-(T_i^{-1})^{-}\bar z_i\leq x\leq(T_i^{-1})^{+}\bar z_i-(T_i^{-1})^{-}\underline z_i\).


\subsection{Distributed Neural Adaptation}
\label{subsec:adaptation}

The adaptive laws are obtained directly from the Lyapunov derivative. No instantaneous neural loss or backpropagation criterion is introduced. For the cooperative realization constructed in Subsection~\ref{subsec:cooperative_realization}, let \(\mathcal M_D\) be Hurwitz and Metzler. Then there exists a diagonal matrix \(\mathcal P=\operatorname{blkdiag}\{P_1,\ldots,P_N\}>0\) and a matrix \(\mathcal Q=\mathcal Q^\top>0\) satisfying \(\mathcal M_D^\top\mathcal P+\mathcal P\mathcal M_D=-\mathcal Q\).

\begin{assumption}
\label{ass:distributed_matching}
For each \(i\in\mathcal V\), there exist matrices \(H_{i+},H_{i-}\in\mathbb R^{p_i\times r}\) such that \(P_iB_i^{+}=C_{zi}^{\top}H_{i+}\) and \(P_iB_i^{-}=C_{zi}^{\top}H_{i-}\).
\end{assumption}

\begin{remark}
\label{rem:distributed_matching}
For \(s\in\{+,-\}\), Assumption~\ref{ass:distributed_matching} is equivalent to \(\operatorname{im}(P_iB_i^{s})\subseteq\operatorname{im}(C_{zi}^{\top})\), or, equivalently, \(\bigl[I-C_{zi}^{\top}(C_{zi}^{\top})^{\dagger}\bigr]P_iB_i^{s}=0\). Whenever this condition holds, one admissible choice is \(H_{is}=(C_{zi}^{\top})^{\dagger}P_iB_i^{s}\). Since \(P_i>0\), one further has \(\operatorname{rank}(B_i^{s})=\operatorname{rank}(P_iB_i^{s})\le\operatorname{rank}(C_{zi})\le p_i\). Hence, Assumption~\ref{ass:distributed_matching} imposes a structural compatibility condition between the uncertainty channels \(B_i^{\pm}\) and the local output subspace \(\operatorname{im}(C_{zi}^{\top})\), and is independent of collective detectability.
\end{remark}

Since \(\bar r_i=C_{zi}\bar\epsilon_i\) and \(\underline r_i=C_{zi}\underline\epsilon_i\), the matching relation above gives \(\bar\epsilon_i^\top P_iB_i^{+}=\bar r_i^\top H_{i+}\), \(\bar\epsilon_i^\top P_iB_i^{-}=\bar r_i^\top H_{i-}\), \(\underline\epsilon_i^\top P_iB_i^{+}=\underline r_i^\top H_{i+}\), and \(\underline\epsilon_i^\top P_iB_i^{-}=\underline r_i^\top H_{i-}\).

Define the measurable adaptive signals \(\xi_{i+}:=H_{i+}^{\top}\bar r_i+H_{i-}^{\top}\underline r_i\) and \(\xi_{i-}:=-H_{i-}^{\top}\bar r_i-H_{i+}^{\top}\underline r_i\).

Consider the composite Lyapunov function \(V:=\frac{1}{2}\bar\epsilon^\top\mathcal P\bar\epsilon+\frac{1}{2}\underline\epsilon^\top\mathcal P\underline\epsilon+\frac{1}{2}\sum_{i=1}^{N}\operatorname{tr}\!\left(\tilde{\bar W}_i^\top\Gamma_{i+}^{-1}\tilde{\bar W}_i\right)+\frac{1}{2}\sum_{i=1}^{N}\operatorname{tr}\!\left(\tilde{\underline W}_i^\top\Gamma_{i-}^{-1}\tilde{\underline W}_i\right)\), where \(\Gamma_{i+}=\Gamma_{i+}^\top>0\) and \(\Gamma_{i-}=\Gamma_{i-}^\top>0\). Since the ideal weights are constant, \(\dot{\tilde{\bar W}}_i=\dot{\bar W}_i\) and \(\dot{\tilde{\underline W}}_i=\dot{\underline W}_i\). Hence, \(\dot V=\bar\epsilon^\top\mathcal P\dot{\bar\epsilon}+\underline\epsilon^\top\mathcal P\dot{\underline\epsilon}+\sum_{i=1}^{N}\operatorname{tr}\!\left(\tilde{\bar W}_i^\top\Gamma_{i+}^{-1}\dot{\bar W}_i\right)+\sum_{i=1}^{N}\operatorname{tr}\!\left(\tilde{\underline W}_i^\top\Gamma_{i-}^{-1}\dot{\underline W}_i\right)\).

Using the stacked interval-error dynamics in~(6), \(\dot{\bar\epsilon}=\mathcal M_D\bar\epsilon+\mathcal B_{+}\bar\eta-\mathcal B_{-}\underline\eta+\bar\Delta\) and \(\dot{\underline\epsilon}=\mathcal M_D\underline\epsilon+\mathcal B_{-}\bar\eta-\mathcal B_{+}\underline\eta+\underline\Delta\), we obtain \(\dot V=\bar\epsilon^\top\mathcal P\left(\mathcal M_D\bar\epsilon+\mathcal B_{+}\bar\eta-\mathcal B_{-}\underline\eta+\bar\Delta\right)+\underline\epsilon^\top\mathcal P\left(\mathcal M_D\underline\epsilon+\mathcal B_{-}\bar\eta-\mathcal B_{+}\underline\eta+\underline\Delta\right)+\sum_{i=1}^{N}\operatorname{tr}\!\left(\tilde{\bar W}_i^\top\Gamma_{i+}^{-1}\dot{\bar W}_i\right)+\sum_{i=1}^{N}\operatorname{tr}\!\left(\tilde{\underline W}_i^\top\Gamma_{i-}^{-1}\dot{\underline W}_i\right)\).

Expanding the preceding expression gives \(\dot V=\bar\epsilon^\top\mathcal P\mathcal M_D\bar\epsilon+\underline\epsilon^\top\mathcal P\mathcal M_D\underline\epsilon+\bar\epsilon^\top\mathcal P\bar\Delta+\underline\epsilon^\top\mathcal P\underline\Delta+\bar\epsilon^\top\mathcal P\mathcal B_{+}\bar\eta+\underline\epsilon^\top\mathcal P\mathcal B_{-}\bar\eta-\bar\epsilon^\top\mathcal P\mathcal B_{-}\underline\eta-\underline\epsilon^\top\mathcal P\mathcal B_{+}\underline\eta+\sum_{i=1}^{N}\operatorname{tr}\!\left(\tilde{\bar W}_i^\top\Gamma_{i+}^{-1}\dot{\bar W}_i\right)+\sum_{i=1}^{N}\operatorname{tr}\!\left(\tilde{\underline W}_i^\top\Gamma_{i-}^{-1}\dot{\underline W}_i\right)\).
Since
\(\mathcal M_D^\top\mathcal P+\mathcal P\mathcal M_D=-Q\),
the first two terms satisfy
\(\bar\epsilon^\top\mathcal P\mathcal M_D\bar\epsilon
=-\frac{1}{2}\bar\epsilon^\top Q\bar\epsilon\) and
\(\underline\epsilon^\top\mathcal P\mathcal M_D\underline\epsilon
=-\frac{1}{2}\underline\epsilon^\top Q\underline\epsilon\).

Moreover, using the block structures of \(\mathcal P\), \(\mathcal B_{\pm}\), \(\bar\eta\), and \(\underline\eta\), the neural terms become \(\bar\epsilon^\top\mathcal P\mathcal B_{+}\bar\eta+\underline\epsilon^\top\mathcal P\mathcal B_{-}\bar\eta=\sum_{i=1}^{N}\left(\bar\epsilon_i^\top P_iB_i^{+}\tilde{\bar W}_i^\top\bar\phi_i+\underline\epsilon_i^\top P_iB_i^{-}\tilde{\bar W}_i^\top\bar\phi_i\right)\), and \(-\bar\epsilon^\top\mathcal P\mathcal B_{-}\underline\eta-\underline\epsilon^\top\mathcal P\mathcal B_{+}\underline\eta=-\sum_{i=1}^{N}\left(\bar\epsilon_i^\top P_iB_i^{-}\tilde{\underline W}_i^\top\underline\phi_i+\underline\epsilon_i^\top P_iB_i^{+}\tilde{\underline W}_i^\top\underline\phi_i\right)\).

For the terms associated with \(\tilde{\bar W}_i\), use \(P_iB_i^{+}=C_{zi}^\top H_{i+}\) and \(P_iB_i^{-}=C_{zi}^\top H_{i-}\) to obtain \(\bar\epsilon_i^\top P_iB_i^{+}\tilde{\bar W}_i^\top\bar\phi_i+\underline\epsilon_i^\top P_iB_i^{-}\tilde{\bar W}_i^\top\bar\phi_i=\bar\epsilon_i^\top C_{zi}^\top H_{i+}\tilde{\bar W}_i^\top\bar\phi_i+\underline\epsilon_i^\top C_{zi}^\top H_{i-}\tilde{\bar W}_i^\top\bar\phi_i\). Since \(\bar r_i=C_{zi}\bar\epsilon_i\) and \(\underline r_i=C_{zi}\underline\epsilon_i\), \(\bar\epsilon_i^\top C_{zi}^\top H_{i+}\tilde{\bar W}_i^\top\bar\phi_i+\underline\epsilon_i^\top C_{zi}^\top H_{i-}\tilde{\bar W}_i^\top\bar\phi_i=\bar r_i^\top H_{i+}\tilde{\bar W}_i^\top\bar\phi_i+\underline r_i^\top H_{i-}\tilde{\bar W}_i^\top\bar\phi_i=\left(\bar r_i^\top H_{i+}+\underline r_i^\top H_{i-}\right)\tilde{\bar W}_i^\top\bar\phi_i=\left(H_{i+}^\top\bar r_i+H_{i-}^\top\underline r_i\right)^\top\tilde{\bar W}_i^\top\bar\phi_i\). Define \(\xi_{i+}:=H_{i+}^\top\bar r_i+H_{i-}^\top\underline r_i\). Then \(\bar\epsilon_i^\top P_iB_i^{+}\tilde{\bar W}_i^\top\bar\phi_i+\underline\epsilon_i^\top P_iB_i^{-}\tilde{\bar W}_i^\top\bar\phi_i=\xi_{i+}^\top\tilde{\bar W}_i^\top\bar\phi_i=\operatorname{tr}\!\left(\tilde{\bar W}_i^\top\bar\phi_i\xi_{i+}^\top\right)\).

Cancellation of this term requires \(\operatorname{tr}\!\left(\tilde{\bar W}_i^\top\Gamma_{i+}^{-1}\dot{\bar W}_i^{\,c}\right)=-\operatorname{tr}\!\left(\tilde{\bar W}_i^\top\bar\phi_i\xi_{i+}^\top\right)\), which is achieved by \(\dot{\bar W}_i^{\,c}=-\Gamma_{i+}\bar\phi_i\xi_{i+}^\top\). To suppress parameter drift, add the leakage term \(-\sigma_{i+}\Gamma_{i+}\bar W_i\), with \(\sigma_{i+}>0\), and obtain
\begin{equation}
\label{update1}
\dot{\bar W}_i
=
-\Gamma_{i+}\bar\phi_i\xi_{i+}^\top
-
\sigma_{i+}\Gamma_{i+}\bar W_i.
\end{equation}

Substitution gives \(\operatorname{tr}\!\left(\tilde{\bar W}_i^\top\Gamma_{i+}^{-1}\dot{\bar W}_i\right)=-\operatorname{tr}\!\left(\tilde{\bar W}_i^\top\bar\phi_i\xi_{i+}^\top\right)-\sigma_{i+}\operatorname{tr}\!\left(\tilde{\bar W}_i^\top\bar W_i\right)\).

For the terms associated with \(\tilde{\underline W}_i\), \(-\bar\epsilon_i^\top P_iB_i^{-}\tilde{\underline W}_i^\top\underline\phi_i-\underline\epsilon_i^\top P_iB_i^{+}\tilde{\underline W}_i^\top\underline\phi_i=-\bar\epsilon_i^\top C_{zi}^\top H_{i-}\tilde{\underline W}_i^\top\underline\phi_i-\underline\epsilon_i^\top C_{zi}^\top H_{i+}\tilde{\underline W}_i^\top\underline\phi_i=-\bar r_i^\top H_{i-}\tilde{\underline W}_i^\top\underline\phi_i-\underline r_i^\top H_{i+}\tilde{\underline W}_i^\top\underline\phi_i=\left(-\bar r_i^\top H_{i-}-\underline r_i^\top H_{i+}\right)\tilde{\underline W}_i^\top\underline\phi_i=\left(-H_{i-}^\top\bar r_i-H_{i+}^\top\underline r_i\right)^\top\tilde{\underline W}_i^\top\underline\phi_i\). Define \(\xi_{i-}:=-H_{i-}^\top\bar r_i-H_{i+}^\top\underline r_i\). Therefore, \(-\bar\epsilon_i^\top P_iB_i^{-}\tilde{\underline W}_i^\top\underline\phi_i-\underline\epsilon_i^\top P_iB_i^{+}\tilde{\underline W}_i^\top\underline\phi_i=\xi_{i-}^\top\tilde{\underline W}_i^\top\underline\phi_i=\operatorname{tr}\!\left(\tilde{\underline W}_i^\top\underline\phi_i\xi_{i-}^\top\right)\).

Thus, \(\dot{\underline W}_i^{\,c}=-\Gamma_{i-}\underline\phi_i\xi_{i-}^\top\). Adding \(-\sigma_{i-}\Gamma_{i-}\underline W_i\), with \(\sigma_{i-}>0\), gives
\begin{equation}
\label{update2}
\dot{\underline W}_i
=
-\Gamma_{i-}\underline\phi_i\xi_{i-}^\top
-
\sigma_{i-}\Gamma_{i-}\underline W_i.
\end{equation}


To remove the persistent-excitation requirement from the online regressor, we augment the instantaneous adaptation with a finite experience-replay stack. For each node \(i\), let \(\mathcal H_i=\{(\Theta_{i,k},Y_{i,k})\}_{k=1}^{M_i}\) denote a set of previously recorded informative data, where \(\Theta_{i,k}:=\int_{t_k-T_I}^{t_k}\phi_i(\tau)\,d\tau\) and \(Y_{i,k}:=z_i(t_k)-z_i(t_k-T_I)-\int_{t_k-T_I}^{t_k}\mathcal F_i(\tau)\,d\tau\), with \(\mathcal F_i\) collecting all known terms of the transformed dynamics. Under the neural representation, \(Y_{i,k}=W_i^{*\top}\Theta_{i,k}+\varepsilon_{i,k}\), where \(\varepsilon_{i,k}\) denotes the integral approximation residual. Define the replay information matrix \(\mathcal G_i:=\frac{1}{M_i}\sum_{k=1}^{M_i}\Theta_{i,k}\Theta_{i,k}^{\top}\). Instead of requiring persistent excitation of \(\phi_i(t)\) for all \(t\), it is sufficient that the stored data satisfy the finite-rank condition \(\lambda_{\min}(\mathcal G_i)\geq \underline g_i>0\). Once such a stack has been collected, the informative samples are retained and repeatedly reused in the adaptive law.

For the lower-weight estimate, define the replay residual \(\underline{\epsilon}_{i,k}:=\Theta_{i,k}^{\top}\underline W_i-Y_{i,k}\), and for the upper-weight estimate define \(\bar{\epsilon}_{i,k}:=\Theta_{i,k}^{\top}\bar W_i-Y_{i,k}\). The corresponding composite experience-replay corrections are \(\underline{\mathcal C}_i:=\frac{1}{M_i}\sum_{k=1}^{M_i}\Theta_{i,k}\underline{\epsilon}_{i,k}^{\top}\) and \(\bar{\mathcal C}_i:=\frac{1}{M_i}\sum_{k=1}^{M_i}\Theta_{i,k}\bar{\epsilon}_{i,k}^{\top}\). Therefore, the pure online laws in (9)--(10) are replaced by
\begin{equation}
\dot{\bar W}_i
=
-\Gamma_{i+}\phi_i\xi_{i+}^{\top}
-\Gamma_{i+}k_{r,i}\bar{\mathcal C}_i,
\label{eq:upper_composite_er}
\end{equation}
and
\begin{equation}
\dot{\underline W}_i
=
-\Gamma_{i-}\phi_i\xi_{i-}^{\top}
-\Gamma_{i-}k_{r,i}\underline{\mathcal C}_i,
\label{eq:lower_composite_er}
\end{equation}
where \(k_{r,i}>0\) is the replay gain. Since \(Y_{i,k}=W_i^{*\top}\Theta_{i,k}\) in the ideal approximation case, \(\bar{\epsilon}_{i,k}=\Theta_{i,k}^{\top}\bar{\tilde W}_i\) and \(\underline{\epsilon}_{i,k}=\Theta_{i,k}^{\top}\underline{\tilde W}_i\), so that \(\bar{\mathcal C}_i=\mathcal G_i\bar{\tilde W}_i\) and \(\underline{\mathcal C}_i=\mathcal G_i\underline{\tilde W}_i\). Consequently, the replay terms generate the negative-definite contributions \(-k_{r,i}\operatorname{tr}\!\left(\bar{\tilde W}_i^{\top}\mathcal G_i\bar{\tilde W}_i\right)\) and \(-k_{r,i}\operatorname{tr}\!\left(\underline{\tilde W}_i^{\top}\mathcal G_i\underline{\tilde W}_i\right)\), respectively, with \(\operatorname{tr}\!\left(\tilde W_i^{\top}\mathcal G_i\tilde W_i\right)\geq \underline g_i\|\tilde W_i\|_F^2\). Hence, parameter convergence is enforced by a finite informative replay stack rather than by persistent excitation of the online trajectory. In the presence of bounded approximation residuals \(\varepsilon_{i,k}\), the same mechanism yields convergence of the neural-weight errors to a residual neighborhood whose size is determined by the replay-data approximation error.

\makeatletter
\@ifundefined{theorem}{\newtheorem{theorem}{Theorem}}{}
\makeatother




\section{Stability and Interval-Enclosure Analysis}
\label{sec:stability_analysis}


\subsection{Boundedness Analysis}
\label{subsec:boundedness}

We first establish boundedness independently of interval preservation. Let \(\bar f_{\Omega}:=\sup_{(x,u)\in\Omega_x\times\Omega_u}\|f(x,u)\|<\infty\) and \(\bar d_{\Omega}:=\sup_{t\geq0}\max\{\|\underline d(t)\|,\|\bar d(t)\|\}<\infty\). By Assumptions~\ref{ass:disturbance}--\ref{ass:neural_interval}, all quantities entering the ideal neural and disturbance residuals are bounded on the prescribed operating region. For each \(i\in\mathcal V\), define \(\Delta_{fi}^{M}:=\max\{\|\bar W_i^{*}\|_{F}\bar\phi_i^{u}+\|\bar\delta_i\|+\bar f_{\Omega},\|\underline W_i^{*}\|_{F}\bar\phi_i^{\ell}+\|\underline\delta_i\|+\bar f_{\Omega}\}\). Further, let \(\Delta_{di}^{M}:=\sup_{t\geq0}\max\{\|\bar q_i-E_i d\|,\|E_i d-\underline q_i\|\}\). Assumption~\ref{ass:disturbance} and Lemma~\ref{lem:matrix_interval} imply \(0\leq\Delta_{di}^{M}<\infty\). Define \(\Delta_i^{M}:=(\|B_i^{+}\|+\|B_i^{-}\|)\Delta_{fi}^{M}+\Delta_{di}^{M}\) and \(\Delta_D:=(\sum_{i=1}^{N}(\Delta_i^{M})^2)^{1/2}\).

\begin{lemma}
\label{lem:distributed_residual_bound}
Under Assumptions~\ref{ass:disturbance}--\ref{ass:neural_interval}, the residuals in \eqref{eq:stacked_interval_error} satisfy \(\|\bar\Delta_i(t)\|\leq\Delta_i^{M}\), \(\|\underline\Delta_i(t)\|\leq\Delta_i^{M}\), and consequently \(\|\bar\Delta(t)\|\leq\Delta_D\) and \(\|\underline\Delta(t)\|\leq\Delta_D\).
\end{lemma}

\begin{IEEEproof}
From \(\bar\Delta_{fi}=\bar W_i^{*\top}\bar\phi_i+\bar\delta_i-f(x,u)\) and \(\underline\Delta_{fi}=f(x,u)-\underline W_i^{*\top}\underline\phi_i+\underline\delta_i\), one obtains \(\|\bar\Delta_{fi}\|,\|\underline\Delta_{fi}\|\leq\Delta_{fi}^{M}\). Using \(\bar\Delta_i=B_i^{+}\bar\Delta_{fi}+B_i^{-}\underline\Delta_{fi}+\bar\Delta_{di}\) and \(\underline\Delta_i=B_i^{+}\underline\Delta_{fi}+B_i^{-}\bar\Delta_{fi}+\underline\Delta_{di}\) gives \(\|\bar\Delta_i\|,\|\underline\Delta_i\|\leq\Delta_i^{M}\). Stacking the local bounds yields \(\|\bar\Delta\|\leq\Delta_D\) and \(\|\underline\Delta\|\leq\Delta_D\).
\end{IEEEproof}


For the composite integral concurrent-learning terms, define the stored integral regressors \(\bar\Theta_{i,k}:=\int_{t_k-T_I}^{t_k}\bar\phi_i(\tau)\,d\tau\) and \(\underline\Theta_{i,k}:=\int_{t_k-T_I}^{t_k}\underline\phi_i(\tau)\,d\tau\), together with the integral regression outputs \(\bar Y_{i,k}:=\bar W_i^{*\top}\bar\Theta_{i,k}+\bar\varepsilon_{i,k}\) and \(\underline Y_{i,k}:=\underline W_i^{*\top}\underline\Theta_{i,k}+\underline\varepsilon_{i,k}\), where \(\bar\varepsilon_{i,k}\) and \(\underline\varepsilon_{i,k}\) collect the bounded integral approximation residuals. Define \(\bar e_{i,k}:=\bar W_i^\top\bar\Theta_{i,k}-\bar Y_{i,k}\) and \(\underline e_{i,k}:=\underline W_i^\top\underline\Theta_{i,k}-\underline Y_{i,k}\), and let \(\bar{\mathcal C}_i:=M_i^{-1}\sum_{k=1}^{M_i}\bar\Theta_{i,k}\bar e_{i,k}^{\top}\) and \(\underline{\mathcal C}_i:=M_i^{-1}\sum_{k=1}^{M_i}\underline\Theta_{i,k}\underline e_{i,k}^{\top}\). Further define the finite-data information matrices \(\bar{\mathcal G}_i:=M_i^{-1}\sum_{k=1}^{M_i}\bar\Theta_{i,k}\bar\Theta_{i,k}^{\top}\) and \(\underline{\mathcal G}_i:=M_i^{-1}\sum_{k=1}^{M_i}\underline\Theta_{i,k}\underline\Theta_{i,k}^{\top}\). The finite-data richness condition requires only \(\lambda_{\min}(\bar{\mathcal G}_i)\geq\underline g_{i+}>0\) and \(\lambda_{\min}(\underline{\mathcal G}_i)\geq\underline g_{i-}>0\) after a finite collection interval and therefore does not require persistent excitation of the online trajectory. Since \(\bar W_i=\bar W_i^*+\bar{\tilde W}_i\) and \(\underline W_i=\underline W_i^*+\underline{\tilde W}_i\), one has \(\bar e_{i,k}=\bar{\tilde W}_i^\top\bar\Theta_{i,k}-\bar\varepsilon_{i,k}\) and \(\underline e_{i,k}=\underline{\tilde W}_i^\top\underline\Theta_{i,k}-\underline\varepsilon_{i,k}\), and hence \(\bar{\mathcal C}_i=\bar{\mathcal G}_i\bar{\tilde W}_i-\bar{\mathcal E}_i\) and \(\underline{\mathcal C}_i=\underline{\mathcal G}_i\underline{\tilde W}_i-\underline{\mathcal E}_i\), where \(\bar{\mathcal E}_i:=M_i^{-1}\sum_{k=1}^{M_i}\bar\Theta_{i,k}\bar\varepsilon_{i,k}^{\top}\) and \(\underline{\mathcal E}_i:=M_i^{-1}\sum_{k=1}^{M_i}\underline\Theta_{i,k}\underline\varepsilon_{i,k}^{\top}\). Let \(\|\bar{\mathcal E}_i\|_F\leq\bar\rho_i\) and \(\|\underline{\mathcal E}_i\|_F\leq\underline\rho_i\).


\begin{theorem}
\label{thm:distributed_uub}
Suppose Assumptions~\ref{ass:disturbance}--\ref{ass:communication}, \ref{ass:distributed_matching}, and \ref{ass:cooperative_realization} hold, and suppose the finite replay stacks satisfy \(\bar{\mathcal G}_i\succeq\underline g_{i+}I\) and \(\underline{\mathcal G}_i\succeq\underline g_{i-}I\) for all \(i\in\mathcal V\). Under the observer \eqref{eq:distributed_interval_observer} and the composite ICL adaptive laws \eqref{update1}--\eqref{update2}, the transformed estimation errors \(\bar\epsilon,\underline\epsilon\) and all neural weight errors \(\bar{\tilde W}_i,\underline{\tilde W}_i\) are uniformly ultimately bounded. Consequently, the transformed interval widths, the physical estimation errors, and the physical interval widths are bounded. Moreover, in the ideal residual-free case with \(\bar\Delta=\underline\Delta=0\), \(\bar{\mathcal E}_i=\underline{\mathcal E}_i=0\), and \(\sigma_{i+}=\sigma_{i-}=0\), the finite replay richness condition is sufficient to drive the neural weight errors to zero without persistent excitation of the online regressors.
\end{theorem}

\begin{IEEEproof}

Consider the composite Lyapunov function introduced in the adaptation analysis, \(V:=\frac{1}{2}\bar\epsilon^\top\mathcal P\bar\epsilon+\frac{1}{2}\underline\epsilon^\top\mathcal P\underline\epsilon+\frac{1}{2}\sum_{i=1}^{N}\operatorname{tr}\!\left(\bar{\tilde W}_i^\top\Gamma_{i+}^{-1}\bar{\tilde W}_i\right)+\frac{1}{2}\sum_{i=1}^{N}\operatorname{tr}\!\left(\underline{\tilde W}_i^\top\Gamma_{i-}^{-1}\underline{\tilde W}_i\right)\). Since the ideal weights are constant, \(\dot{\bar{\tilde W}}_i=\dot{\bar W}_i\) and \(\dot{\underline{\tilde W}}_i=\dot{\underline W}_i\). Hence, \(\dot V=\bar\epsilon^\top\mathcal P\dot{\bar\epsilon}+\underline\epsilon^\top\mathcal P\dot{\underline\epsilon}+\sum_{i=1}^{N}\operatorname{tr}\!\left(\bar{\tilde W}_i^\top\Gamma_{i+}^{-1}\dot{\bar W}_i\right)+\sum_{i=1}^{N}\operatorname{tr}\!\left(\underline{\tilde W}_i^\top\Gamma_{i-}^{-1}\dot{\underline W}_i\right)\).

Using the stacked interval-error dynamics \eqref{eq:stacked_interval_error}, \(\dot{\bar\epsilon}=\mathcal M_D\bar\epsilon+\mathcal B_{+}\bar\eta-\mathcal B_{-}\underline\eta+\bar\Delta\) and \(\dot{\underline\epsilon}=\mathcal M_D\underline\epsilon+\mathcal B_{-}\bar\eta-\mathcal B_{+}\underline\eta+\underline\Delta\), we obtain \(\dot V=\bar\epsilon^\top\mathcal P\mathcal M_D\bar\epsilon+\underline\epsilon^\top\mathcal P\mathcal M_D\underline\epsilon+\bar\epsilon^\top\mathcal P\bar\Delta+\underline\epsilon^\top\mathcal P\underline\Delta+\bar\epsilon^\top\mathcal P\mathcal B_{+}\bar\eta+\underline\epsilon^\top\mathcal P\mathcal B_{-}\bar\eta-\bar\epsilon^\top\mathcal P\mathcal B_{-}\underline\eta-\underline\epsilon^\top\mathcal P\mathcal B_{+}\underline\eta+\sum_{i=1}^{N}\operatorname{tr}\!\left(\bar{\tilde W}_i^\top\Gamma_{i+}^{-1}\dot{\bar W}_i\right)+\sum_{i=1}^{N}\operatorname{tr}\!\left(\underline{\tilde W}_i^\top\Gamma_{i-}^{-1}\dot{\underline W}_i\right)\).

By \(\mathcal M_D^\top\mathcal P+\mathcal P\mathcal M_D=-\mathcal Q\), the linear-error terms satisfy \(\bar\epsilon^\top\mathcal P\mathcal M_D\bar\epsilon=-\frac{1}{2}\bar\epsilon^\top\mathcal Q\bar\epsilon\) and \(\underline\epsilon^\top\mathcal P\mathcal M_D\underline\epsilon=-\frac{1}{2}\underline\epsilon^\top\mathcal Q\underline\epsilon\).

From the adaptive-law construction, \(\bar\epsilon^\top\mathcal P\mathcal B_{+}\bar\eta+\underline\epsilon^\top\mathcal P\mathcal B_{-}\bar\eta=\sum_{i=1}^{N}\operatorname{tr}\!\left(\bar{\tilde W}_i^\top\bar\phi_i\xi_{i+}^\top\right)\), whereas \(-\bar\epsilon^\top\mathcal P\mathcal B_{-}\underline\eta-\underline\epsilon^\top\mathcal P\mathcal B_{+}\underline\eta=\sum_{i=1}^{N}\operatorname{tr}\!\left(\underline{\tilde W}_i^\top\underline\phi_i\xi_{i-}^\top\right)\).

The composite ICL adaptive laws give \(\operatorname{tr}\!\left(\bar{\tilde W}_i^\top\Gamma_{i+}^{-1}\dot{\bar W}_i\right)=-\operatorname{tr}\!\left(\bar{\tilde W}_i^\top\bar\phi_i\xi_{i+}^\top\right)-k_{r,i+}\operatorname{tr}\!\left(\bar{\tilde W}_i^\top\bar{\mathcal C}_i\right)-\sigma_{i+}\operatorname{tr}\!\left(\bar{\tilde W}_i^\top\bar W_i\right)\) and \(\operatorname{tr}\!\left(\underline{\tilde W}_i^\top\Gamma_{i-}^{-1}\dot{\underline W}_i\right)=-\operatorname{tr}\!\left(\underline{\tilde W}_i^\top\underline\phi_i\xi_{i-}^\top\right)-k_{r,i-}\operatorname{tr}\!\left(\underline{\tilde W}_i^\top\underline{\mathcal C}_i\right)-\sigma_{i-}\operatorname{tr}\!\left(\underline{\tilde W}_i^\top\underline W_i\right)\). Thus, all neural state--weight cross terms cancel exactly, while the ICL terms remain as additional negative finite-data corrections, yielding

\begin{equation}
\begin{aligned}
\dot V
={}&
-\frac{1}{2}\bar\epsilon^\top\mathcal Q\bar\epsilon
-\frac{1}{2}\underline\epsilon^\top\mathcal Q\underline\epsilon
+\bar\epsilon^\top\mathcal P\bar\Delta
+\underline\epsilon^\top\mathcal P\underline\Delta
\\
&-
\sum_{i=1}^{N}
k_{r,i+}
\operatorname{tr}\!\left(
\bar{\tilde W}_i^\top
\bar{\mathcal C}_i
\right)
-
\sum_{i=1}^{N}
k_{r,i-}
\operatorname{tr}\!\left(
\underline{\tilde W}_i^\top
\underline{\mathcal C}_i
\right)
\\
&-
\sum_{i=1}^{N}
\sigma_{i+}
\operatorname{tr}\!\left(
\bar{\tilde W}_i^\top\bar W_i
\right)
-
\sum_{i=1}^{N}
\sigma_{i-}
\operatorname{tr}\!\left(
\underline{\tilde W}_i^\top\underline W_i
\right).
\end{aligned}
\label{eq:distributed_Vdot}
\end{equation}


For the upper ICL contribution, \(\bar{\mathcal C}_i=\bar{\mathcal G}_i\bar{\tilde W}_i-\bar{\mathcal E}_i\) gives \(-k_{r,i+}\operatorname{tr}\!\left(\bar{\tilde W}_i^\top\bar{\mathcal C}_i\right)=-k_{r,i+}\operatorname{tr}\!\left(\bar{\tilde W}_i^\top\bar{\mathcal G}_i\bar{\tilde W}_i\right)+k_{r,i+}\operatorname{tr}\!\left(\bar{\tilde W}_i^\top\bar{\mathcal E}_i\right)\). Since \(\bar{\mathcal G}_i\succeq\underline g_{i+}I\), \(\operatorname{tr}\!\left(\bar{\tilde W}_i^\top\bar{\mathcal G}_i\bar{\tilde W}_i\right)\geq\underline g_{i+}\|\bar{\tilde W}_i\|_F^2\), whereas the Frobenius Cauchy--Schwarz inequality yields \(\operatorname{tr}\!\left(\bar{\tilde W}_i^\top\bar{\mathcal E}_i\right)\leq\|\bar{\tilde W}_i\|_F\|\bar{\mathcal E}_i\|_F\leq\bar\rho_i\|\bar{\tilde W}_i\|_F\). Using \(2ab\leq a^2+b^2\) with \(a=\sqrt{\underline g_{i+}}\|\bar{\tilde W}_i\|_F\) and \(b=\bar\rho_i/\sqrt{\underline g_{i+}}\) gives \(-k_{r,i+}\operatorname{tr}\!\left(\bar{\tilde W}_i^\top\bar{\mathcal C}_i\right)\leq-\frac{k_{r,i+}\underline g_{i+}}{2}\|\bar{\tilde W}_i\|_F^2+\frac{k_{r,i+}}{2\underline g_{i+}}\bar\rho_i^2\).

Similarly, \(\underline{\mathcal C}_i=\underline{\mathcal G}_i\underline{\tilde W}_i-\underline{\mathcal E}_i\), \(\underline{\mathcal G}_i\succeq\underline g_{i-}I\), and \(\|\underline{\mathcal E}_i\|_F\leq\underline\rho_i\) imply \(-k_{r,i-}\operatorname{tr}\!\left(\underline{\tilde W}_i^\top\underline{\mathcal C}_i\right)\leq-\frac{k_{r,i-}\underline g_{i-}}{2}\|\underline{\tilde W}_i\|_F^2+\frac{k_{r,i-}}{2\underline g_{i-}}\underline\rho_i^2\).

For the upper-weight leakage term, using \(\bar W_i=\bar{\tilde W}_i+\bar W_i^*\) gives \(-\sigma_{i+}\operatorname{tr}\!\left(\bar{\tilde W}_i^\top\bar W_i\right)=-\sigma_{i+}\|\bar{\tilde W}_i\|_F^2-\sigma_{i+}\operatorname{tr}\!\left(\bar{\tilde W}_i^\top\bar W_i^*\right)\). The Frobenius Cauchy--Schwarz inequality gives \(\left|\operatorname{tr}\!\left(\bar{\tilde W}_i^\top\bar W_i^*\right)\right|\leq\|\bar{\tilde W}_i\|_F\|\bar W_i^*\|_F\), and hence \(-\sigma_{i+}\operatorname{tr}\!\left(\bar{\tilde W}_i^\top\bar W_i\right)\leq-\sigma_{i+}\|\bar{\tilde W}_i\|_F^2+\sigma_{i+}\|\bar{\tilde W}_i\|_F\|\bar W_i^*\|_F\). Using \(2ab\leq a^2+b^2\), one obtains \(-\sigma_{i+}\operatorname{tr}\!\left(\bar{\tilde W}_i^\top\bar W_i\right)\leq-\frac{\sigma_{i+}}{2}\|\bar{\tilde W}_i\|_F^2+\frac{\sigma_{i+}}{2}\|\bar W_i^*\|_F^2\).

Similarly, \(\underline W_i=\underline{\tilde W}_i+\underline W_i^*\) gives \(-\sigma_{i-}\operatorname{tr}\!\left(\underline{\tilde W}_i^\top\underline W_i\right)=-\sigma_{i-}\|\underline{\tilde W}_i\|_F^2-\sigma_{i-}\operatorname{tr}\!\left(\underline{\tilde W}_i^\top\underline W_i^*\right)\), with \(\left|\operatorname{tr}\!\left(\underline{\tilde W}_i^\top\underline W_i^*\right)\right|\leq\|\underline{\tilde W}_i\|_F\|\underline W_i^*\|_F\). Therefore, \(-\sigma_{i-}\operatorname{tr}\!\left(\underline{\tilde W}_i^\top\underline W_i\right)\leq-\frac{\sigma_{i-}}{2}\|\underline{\tilde W}_i\|_F^2+\frac{\sigma_{i-}}{2}\|\underline W_i^*\|_F^2\).

Let \(q_D:=\lambda_{\min}(\mathcal Q)>0\) and \(p_D:=\lambda_{\max}(\mathcal P)>0\). By Lemma~\ref{lem:distributed_residual_bound}, \(\bar\epsilon^\top\mathcal P\bar\Delta\leq p_D\|\bar\epsilon\|\Delta_D\) and \(\underline\epsilon^\top\mathcal P\underline\Delta\leq p_D\|\underline\epsilon\|\Delta_D\). Young's inequality gives \(p_D\|\bar\epsilon\|\Delta_D\leq(q_D/4)\|\bar\epsilon\|^2+(p_D^2/q_D)\Delta_D^2\) and \(p_D\|\underline\epsilon\|\Delta_D\leq(q_D/4)\|\underline\epsilon\|^2+(p_D^2/q_D)\Delta_D^2\). Moreover, \(\bar\epsilon^\top\mathcal Q\bar\epsilon\geq q_D\|\bar\epsilon\|^2\) and \(\underline\epsilon^\top\mathcal Q\underline\epsilon\geq q_D\|\underline\epsilon\|^2\). Substitution of the preceding bounds into \eqref{eq:distributed_Vdot} gives

\begin{equation}
\begin{aligned}
\dot V
\leq{}&
-\frac{q_D}{4}
\left(
\|\bar\epsilon\|^2
+
\|\underline\epsilon\|^2
\right)
\\
&-
\frac{1}{2}
\sum_{i=1}^{N}
\left(
\sigma_{i+}
+
k_{r,i+}\underline g_{i+}
\right)
\|\bar{\tilde W}_i\|_F^2
\\
&-
\frac{1}{2}
\sum_{i=1}^{N}
\left(
\sigma_{i-}
+
k_{r,i-}\underline g_{i-}
\right)
\|\underline{\tilde W}_i\|_F^2
+
\Delta_V,
\end{aligned}
\label{eq:Vdot_uub}
\end{equation}

where \(\Delta_V:=2p_D^2\Delta_D^2/q_D+\frac{1}{2}\sum_{i=1}^{N}\left[\sigma_{i+}\|\bar W_i^*\|_F^2+\sigma_{i-}\|\underline W_i^*\|_F^2+\frac{k_{r,i+}}{\underline g_{i+}}\bar\rho_i^2+\frac{k_{r,i-}}{\underline g_{i-}}\underline\rho_i^2\right]\).

Define
\(c_{D,0}:=\frac{q_D}{2p_D}\),
\(c_{D,+}:=\min_i\frac{\sigma_{i+}+k_{r,i+}\underline g_{i+}}{\lambda_{\max}(\Gamma_{i+}^{-1})}\),
and
\(c_{D,-}:=\min_i\frac{\sigma_{i-}+k_{r,i-}\underline g_{i-}}{\lambda_{\max}(\Gamma_{i-}^{-1})}\).
Then
\(c_D:=\min\{c_{D,0},c_{D,+},c_{D,-}\}\).. Since \(\frac{1}{2}\bar\epsilon^\top\mathcal P\bar\epsilon+\frac{1}{2}\underline\epsilon^\top\mathcal P\underline\epsilon\leq\frac{p_D}{2}(\|\bar\epsilon\|^2+\|\underline\epsilon\|^2)\), and \(\frac{1}{2}\operatorname{tr}\!\left(\bar{\tilde W}_i^\top\Gamma_{i+}^{-1}\bar{\tilde W}_i\right)\leq\frac{\lambda_{\max}(\Gamma_{i+}^{-1})}{2}\|\bar{\tilde W}_i\|_F^2\), with the analogous relation for \(\underline{\tilde W}_i\), one has \(c_D>0\), and \eqref{eq:Vdot_uub} implies

\begin{equation}
\dot V\leq-c_DV+\Delta_V.
\label{eq:scalar_uub}
\end{equation}

Multiplying \eqref{eq:scalar_uub} by \(e^{c_Dt}\) gives \(\frac{d}{dt}(e^{c_Dt}V(t))\leq\Delta_Ve^{c_Dt}\). Integration from \(0\) to \(t\) yields \(e^{c_Dt}V(t)-V(0)\leq\frac{\Delta_V}{c_D}(e^{c_Dt}-1)\), and therefore \(V(t)\leq e^{-c_Dt}V(0)+\frac{\Delta_V}{c_D}(1-e^{-c_Dt})\). Consequently, \(\limsup_{t\to\infty}V(t)\leq\Delta_V/c_D\).

Since \(V\geq\frac{\lambda_{\min}(\mathcal P)}{2}(\|\bar\epsilon\|^2+\|\underline\epsilon\|^2)\), it follows that \(\limsup_{t\to\infty}\|\bar\epsilon(t)\|\leq\sqrt{\frac{2\Delta_V}{c_D\lambda_{\min}(\mathcal P)}}\) and \(\limsup_{t\to\infty}\|\underline\epsilon(t)\|\leq\sqrt{\frac{2\Delta_V}{c_D\lambda_{\min}(\mathcal P)}}\). Moreover, \(V\geq\frac{\lambda_{\min}(\Gamma_{i+}^{-1})}{2}\|\bar{\tilde W}_i\|_F^2\) and \(V\geq\frac{\lambda_{\min}(\Gamma_{i-}^{-1})}{2}\|\underline{\tilde W}_i\|_F^2\), so that \(\limsup_{t\to\infty}\|\bar{\tilde W}_i(t)\|_F\leq\sqrt{\frac{2\Delta_V}{c_D\lambda_{\min}(\Gamma_{i+}^{-1})}}\) and \(\limsup_{t\to\infty}\|\underline{\tilde W}_i(t)\|_F\leq\sqrt{\frac{2\Delta_V}{c_D\lambda_{\min}(\Gamma_{i-}^{-1})}}\). Hence the transformed estimation errors and all neural weight errors are uniformly ultimately bounded.

The role of ICL is explicit in \eqref{eq:Vdot_uub}: the quantities \(k_{r,i+}\underline g_{i+}\) and \(k_{r,i-}\underline g_{i-}\) provide additional coercive decay of the neural-weight errors using only a finite informative data stack. Therefore, persistent excitation of \(\bar\phi_i(t)\) and \(\underline\phi_i(t)\) for all \(t\geq0\) is not required. In the ideal replay case \(\bar{\mathcal E}_i=\underline{\mathcal E}_i=0\), and with \(\sigma_{i+}=\sigma_{i-}=0\), the replay contributions reduce exactly to \(-k_{r,i+}\operatorname{tr}\!\left(\bar{\tilde W}_i^\top\bar{\mathcal G}_i\bar{\tilde W}_i\right)\leq-k_{r,i+}\underline g_{i+}\|\bar{\tilde W}_i\|_F^2\) and \(-k_{r,i-}\operatorname{tr}\!\left(\underline{\tilde W}_i^\top\underline{\mathcal G}_i\underline{\tilde W}_i\right)\leq-k_{r,i-}\underline g_{i-}\|\underline{\tilde W}_i\|_F^2\). If, additionally, \(\bar\Delta=\underline\Delta=0\), then \(\Delta_V=0\), so \eqref{eq:scalar_uub} becomes \(\dot V\leq-c_DV\), from which \(V(t)\leq e^{-c_Dt}V(0)\). Consequently, \(\bar{\tilde W}_i(t)\to0\) and \(\underline{\tilde W}_i(t)\to0\) exponentially as \(t\to\infty\), thereby establishing exact neural-weight convergence under finite-data richness without persistent excitation.

The transformed interval width at node \(i\) is \(\omega_i^{z}:=\bar z_i-\underline z_i=\bar\epsilon_i+\underline\epsilon_i\). Therefore, \(\|\omega_i^{z}\|\leq\|\bar\epsilon_i\|+\|\underline\epsilon_i\|\), so the transformed interval width is also bounded.

For the physical-state reconstruction, \(\bar x_i=(T_i^{-1})^{+}\bar z_i-(T_i^{-1})^{-}\underline z_i\) and \(\underline x_i=(T_i^{-1})^{+}\underline z_i-(T_i^{-1})^{-}\bar z_i\). Since \(x=T_i^{-1}z_i\), \(\bar x_i-x=(T_i^{-1})^{+}\bar\epsilon_i+(T_i^{-1})^{-}\underline\epsilon_i\) and \(x-\underline x_i=(T_i^{-1})^{+}\underline\epsilon_i+(T_i^{-1})^{-}\bar\epsilon_i\). Thus, boundedness of the transformed errors implies boundedness of the physical estimation errors. In particular, defining \(\tau_i:=\|(T_i^{-1})^{+}\|+\|(T_i^{-1})^{-}\|\), gives \(\limsup_{t\to\infty}\max\{\|\bar x_i-x\|,\|x-\underline x_i\|\}\leq\tau_i\sqrt{\frac{2\Delta_V}{c_D\lambda_{\min}(\mathcal P)}}\).

Finally, \(\bar x_i-\underline x_i=[(T_i^{-1})^{+}+(T_i^{-1})^{-}](\bar z_i-\underline z_i)=|T_i^{-1}|(\bar z_i-\underline z_i)\), where \(|T_i^{-1}|:=(T_i^{-1})^{+}+(T_i^{-1})^{-}\). Hence the conditioning and sign structure of \(T_i^{-1}\) determine the amplification of the transformed interval width.
\end{IEEEproof}


\subsection{Distributed Interval Preservation}
\label{subsec:interval_preservation}

Theorem~\ref{thm:distributed_uub} establishes boundedness but does not
imply componentwise state enclosure. The latter requires the online
neural models themselves to preserve the nonlinear ordering.

\begin{assumption}
\label{ass:online_neural_enclosure}
For every \(i\in\mathcal V\), whenever
\(\underline z_i(t)\leq z_i(t)\leq\bar z_i(t)\), the online neural
models satisfy
\(\underline f_i(t)\leq f(x(t),u(t))\leq\bar f_i(t)\).
\end{assumption}

Assumption~\ref{ass:online_neural_enclosure} is an interval-validity
condition and is independent of weight convergence. In particular,
boundedness of \(\tilde{\bar W}_i\) and
\(\tilde{\underline W}_i\) does not by itself imply this ordering.

Define
\(q_i^{+}:=
B_i^{+}(\bar f_i-f)+B_i^{-}(f-\underline f_i)
+\bar q_i-E_i d\) and
\(q_i^{-}:=
B_i^{+}(f-\underline f_i)+B_i^{-}(\bar f_i-f)
+E_i d-\underline q_i\). By
Assumptions~\ref{ass:disturbance} and
\ref{ass:online_neural_enclosure},
\(q_i^{+}\geq0\) and \(q_i^{-}\geq0\). Let
\(q^{+}:=\operatorname{col}\{q_1^{+},\ldots,q_N^{+}\}\) and
\(q^{-}:=\operatorname{col}\{q_1^{-},\ldots,q_N^{-}\}\). The complete
upper and lower error systems can then be written as
\begin{equation}
\dot{\bar\epsilon}
=
\mathcal M_D\bar\epsilon+q^{+},
\qquad
\dot{\underline\epsilon}
=
\mathcal M_D\underline\epsilon+q^{-}.
\label{eq:positive_distributed_error}
\end{equation}

\begin{theorem}
\label{thm:distributed_interval}
Suppose Assumptions~\ref{ass:disturbance},
\ref{ass:graph},
\ref{ass:collective_detectability},
\ref{ass:communication},
\ref{ass:cooperative_realization}, and
\ref{ass:online_neural_enclosure} hold. If
\(\bar\epsilon(0)\geq0\) and
\(\underline\epsilon(0)\geq0\), then
\(\bar\epsilon(t)\geq0\) and
\(\underline\epsilon(t)\geq0\) for all \(t\geq0\).
\end{theorem}

\begin{IEEEproof}
By Assumption~\ref{ass:cooperative_realization},
\(\mathcal M_D\) is Metzler; hence
\(e^{\mathcal M_Dt}\geq0\) for all \(t\geq0\). The solutions of
\eqref{eq:positive_distributed_error} are
\(\bar\epsilon(t)=
e^{\mathcal M_Dt}\bar\epsilon(0)
+\int_{0}^{t}e^{\mathcal M_D(t-\tau)}
q^{+}(\tau)\,d\tau\) and
\(\underline\epsilon(t)=
e^{\mathcal M_Dt}\underline\epsilon(0)
+\int_{0}^{t}e^{\mathcal M_D(t-\tau)}
q^{-}(\tau)\,d\tau\).
Every term on the right-hand side is componentwise nonnegative.
Therefore
\(\bar\epsilon(t),\underline\epsilon(t)\geq0\)
for all \(t\geq0\).
\end{IEEEproof}

For the direct realization \(T_i=I_n\), one has \(z_i=x\),
\(\bar\epsilon_i=\bar x_i-x\), and
\(\underline\epsilon_i=x-\underline x_i\).
Theorem~\ref{thm:distributed_interval} therefore gives
\begin{equation}
\underline x_i(t)
\leq
x(t)
\leq
\bar x_i(t),
\qquad
\forall i\in\mathcal V,\quad t\geq0 .
\label{eq:direct_distributed_enclosure}
\end{equation}

The interval width
\(\omega^z:=
\operatorname{col}\{\omega_1^z,\ldots,\omega_N^z\}
=\bar\epsilon+\underline\epsilon\) satisfies
\(\dot\omega^z=\mathcal M_D\omega^z+q^{+}+q^{-}\).
Since \(\mathcal M_D\) is Metzler,
\(q^{+}+q^{-}\geq0\), and
\(\omega^z(0)\geq0\), one also has
\(\omega^z(t)\geq0\). Thus the upper and lower trajectories cannot
cross.

\begin{remark}
\label{rem:uub_not_interval}
Theorem~\ref{thm:distributed_uub} does not require
Assumption~\ref{ass:online_neural_enclosure};
Theorem~\ref{thm:distributed_interval} does. Hence neural-weight
boundedness and interval preservation are distinct properties.
Temporary violation of
\(\underline f_i\leq f\leq\bar f_i\)
does not invalidate the UUB result but removes the positive-system
guarantee in \eqref{eq:positive_distributed_error}.
\end{remark}

\section{Numerical Results}
\label{sec:results}

\subsection{Simulation Setup}
\label{subsec:simulation_setup}

Consider the nonlinear control-affine system
\begin{equation}
    \dot x=F(x)+G(x)u+Ed,
    \qquad
    x=\operatorname{col}\{x_1,x_2,x_3,x_4\}.
\label{eq:sim_original_system}
\end{equation}

The drift \(F(x)=\operatorname{col}\{F_1,F_2,F_3,F_4\}\) is given by
\(F_1=0.2x_1+0.8x_2+0.35\sin x_1+0.12x_2x_3+0.08x_1x_4^2\),
\(F_2=0.3x_1-0.4x_2+0.25\tanh(x_1x_3)+0.16x_2^3
+0.10\sin(x_2x_4)\),
\(F_3=-0.2x_3+0.7x_4+0.28x_1x_2+0.18\sin x_3
+0.07x_1^2x_4\), and
\(F_4=0.4x_3-0.5x_4+0.22[\cos(x_1x_4)-1]
+0.12x_2x_3^2+0.06\tanh(x_1+x_2)\).

The input vector is
\(G(x)=\operatorname{col}\{
0.8+0.15\cos x_2,\,
0.5+0.10x_1^2,\,
0.6+0.12\sin x_3,\,
0.35+0.08x_2^2\}\),
and
\(E=\operatorname{col}\{0.15,-0.10,0.08,0.12\}\).

To obtain the form in \eqref{eq:distributed_plant}, choose
\(A=\operatorname{blkdiag}\{A_1,A_2\}\), with
\(A_1=\left[\begin{smallmatrix}
0.2&0.8\\
0.3&-0.4
\end{smallmatrix}\right]\),
\(A_2=\left[\begin{smallmatrix}
-0.2&0.7\\
0.4&-0.5
\end{smallmatrix}\right]\),
and \(B=I_4\).
Defining \(f(x,u):=F(x)-Ax+G(x)u\) gives
\(\dot x=Ax+Bf(x,u)+Ed\).

The nonlinear residual \(f(x,u)=\operatorname{col}\{f_1,f_2,f_3,f_4\}\)
has components
\(f_1=0.35\sin x_1+0.12x_2x_3+0.08x_1x_4^2
+(0.8+0.15\cos x_2)u\),
\(f_2=0.25\tanh(x_1x_3)+0.16x_2^3+0.10\sin(x_2x_4)
+(0.5+0.10x_1^2)u\),
\(f_3=0.28x_1x_2+0.18\sin x_3+0.07x_1^2x_4
+(0.6+0.12\sin x_3)u\), and
\(f_4=0.22[\cos(x_1x_4)-1]+0.12x_2x_3^2
+0.06\tanh(x_1+x_2)+(0.35+0.08x_2^2)u\).

The known input is
\(u(t)=0.6\sin(0.7t)+0.25\sin(1.6t)\),
while the disturbance is
\(d(t)=0.15\sin(0.8t)+0.05\sin(2.1t)\),
with only the bounds
\(\underline d=-0.2\) and \(\bar d=0.2\)
available to the observers.

The plant is monitored by \(N=5\) scalar-output nodes.
Let \(e_k\) denote the \(k\)-th canonical basis vector of
\(\mathbb R^4\).
The sensing matrices are
\(C_1=C_5=e_1^\top\),
\(C_2=e_2^\top\),
\(C_3=e_3^\top\), and
\(C_4=e_4^\top\).
Hence, each local pair \((A,C_i)\) is nondetectable, whereas
\(C_{\mathcal V}=\operatorname{col}\{C_1,\ldots,C_5\}\)
yields an observable aggregate pair \((A,C_{\mathcal V})\), and thus
collective detectability holds.

The communication topology is the directed ring
\(1\rightarrow2\rightarrow3\rightarrow4\rightarrow5\rightarrow1\),
with unit edge weights, and is strongly connected.

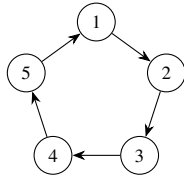
\begin{figure}[H]
\centering
\begin{tikzpicture}[
    scale=0.78,
    >=Stealth,
    node/.style={
        circle,
        draw,
        inner sep=0pt,
        minimum size=5mm,
        font=\scriptsize
    },
    edge/.style={->,thin}
]

\node[node] (n1) at (90:1.25cm)   {1};
\node[node] (n2) at (18:1.25cm)   {2};
\node[node] (n3) at (-54:1.25cm)  {3};
\node[node] (n4) at (-126:1.25cm) {4};
\node[node] (n5) at (162:1.25cm)  {5};

\draw[edge] (n1) -- (n2);
\draw[edge] (n2) -- (n3);
\draw[edge] (n3) -- (n4);
\draw[edge] (n4) -- (n5);
\draw[edge] (n5) -- (n1);

\end{tikzpicture}
\caption{Directed communication network with \(N=5\) observer nodes.}
\label{fig:simulation_graph}
\end{figure}

For the direct cooperative realization, choose
\(T_i=I_4\),
\(K_i=4C_i^\top\), and
\(\kappa=1.5\),
for all \(i\in\mathcal V\).
Since \(A\) is Metzler and each \(K_iC_i\) is diagonal and nonnegative,
the network matrix
\(\mathcal M_D
=I_5\otimes A
-\operatorname{blkdiag}\{K_iC_i\}_{i=1}^{5}
-\kappa\mathcal L\otimes I_4\)
is Metzler.
For the selected topology and gains,
\(\max_{\lambda\in\sigma(\mathcal M_D)}
\operatorname{Re}(\lambda)=-0.2026<0\);
hence, \(\mathcal M_D\) is Hurwitz.

Each node employs \(N_\phi=15\) Gaussian radial basis functions
\(\phi_{ik}(\zeta)
=\exp[-\|\zeta-c_{ik}\|^2/1.2^2]\),
where
\(\zeta=\operatorname{col}\{x,u\}\).
The centers \(c_{ik}\) are distributed over the prescribed operating
region.
The adaptation parameters are
\(\Gamma_{i+}=\Gamma_{i-}=8I_{N_\phi}\) and
\(\sigma_{i+}=\sigma_{i-}=0.03\),
with
\(\bar W_i(0)=\underline W_i(0)=0\).

The initial state is
\(x(0)=\operatorname{col}\{0.6,-0.4,0.5,-0.3\}\).
The initial intervals are
\(\underline x_i(0)=x(0)-\delta_i\mathbf1_4\) and
\(\bar x_i(0)=x(0)+\delta_i\mathbf1_4\),
where
\(\delta_i=0.25+0.03(i-1)\),
\(i=1,\ldots,5\).
The simulation is conducted over \(t\in[0,25]\)~s.


\subsection{Simulation Results}
\label{sec:simulation_results}

\begin{figure}[!ht]
    \centering
    \includegraphics[width=0.90\columnwidth]
    {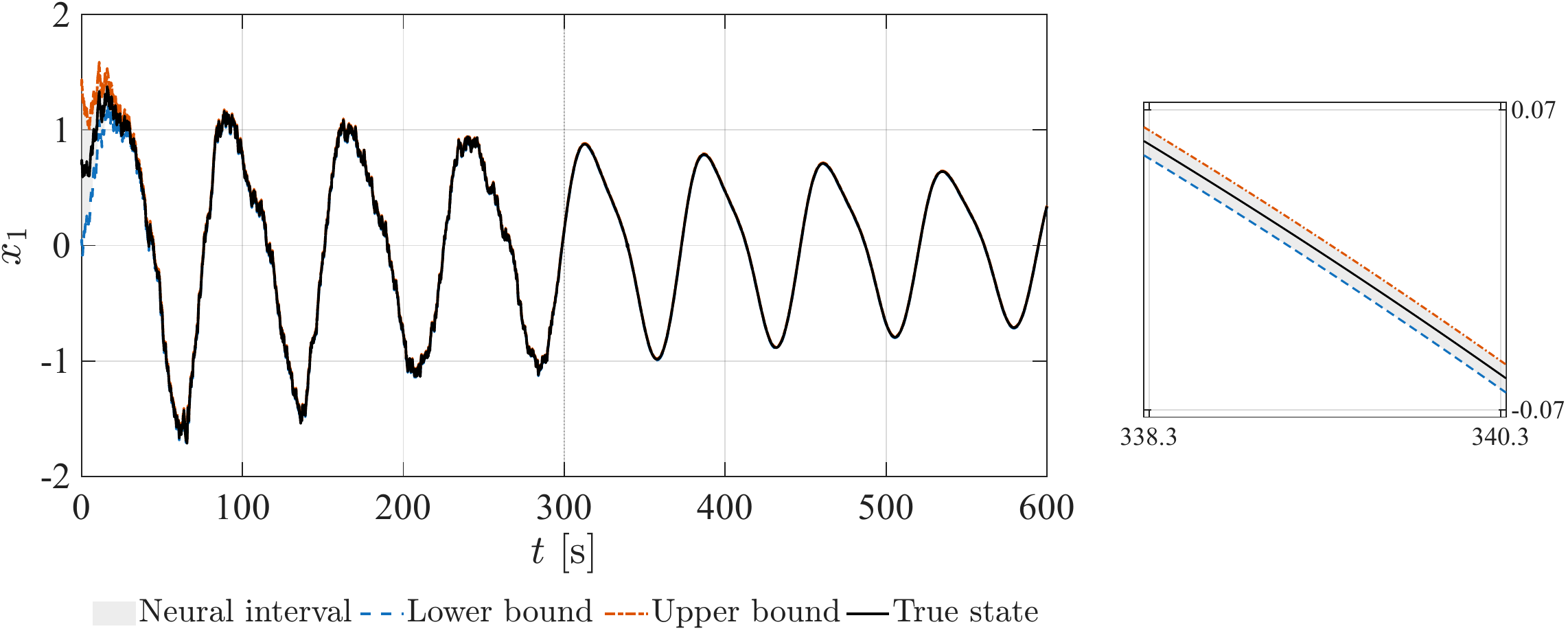}
    \caption{Neural interval estimate for state \(x_1\) at Node~3.}
    \label{fig:neural_interval_x1}
\end{figure}

\begin{figure}[!ht]
    \centering
    \includegraphics[width=0.90\columnwidth]
    {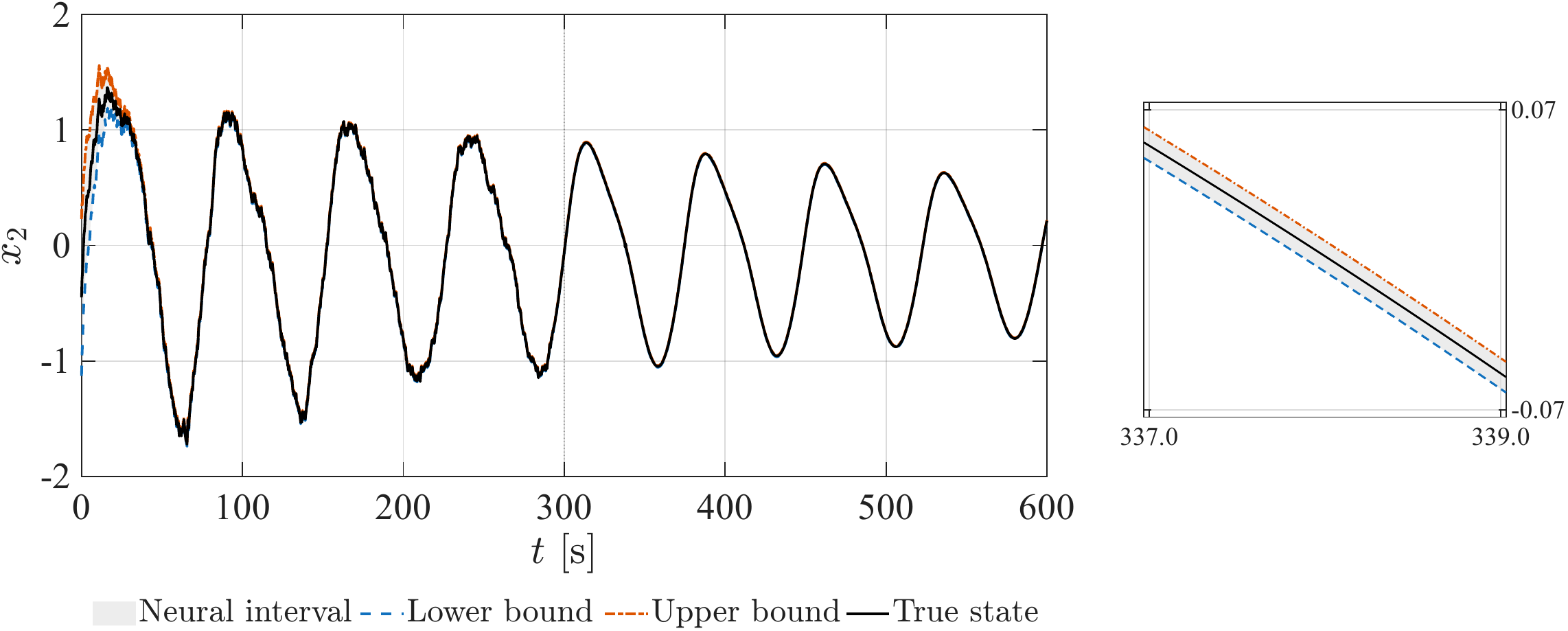}
    \caption{Neural interval estimate for state \(x_2\) at Node~3.}
    \label{fig:neural_interval_x2}
\end{figure}



\begin{figure}[!ht]
    \centering
    \includegraphics[width=\columnwidth]
    {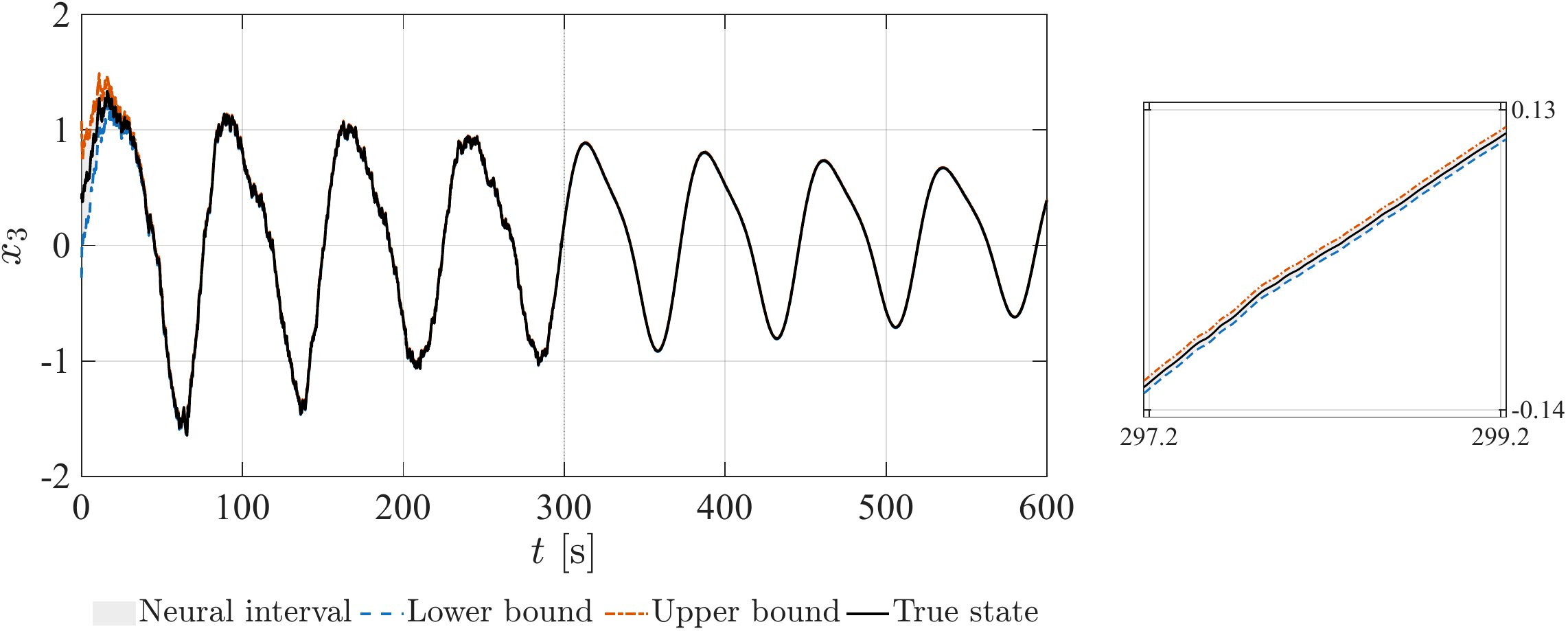}
    \caption{Neural interval estimate for state \(x_3\) at Node~3.}
    \label{fig:neural_interval_x3}
\end{figure}

\begin{figure}[!ht]
    \centering
    \includegraphics[width=\columnwidth]
    {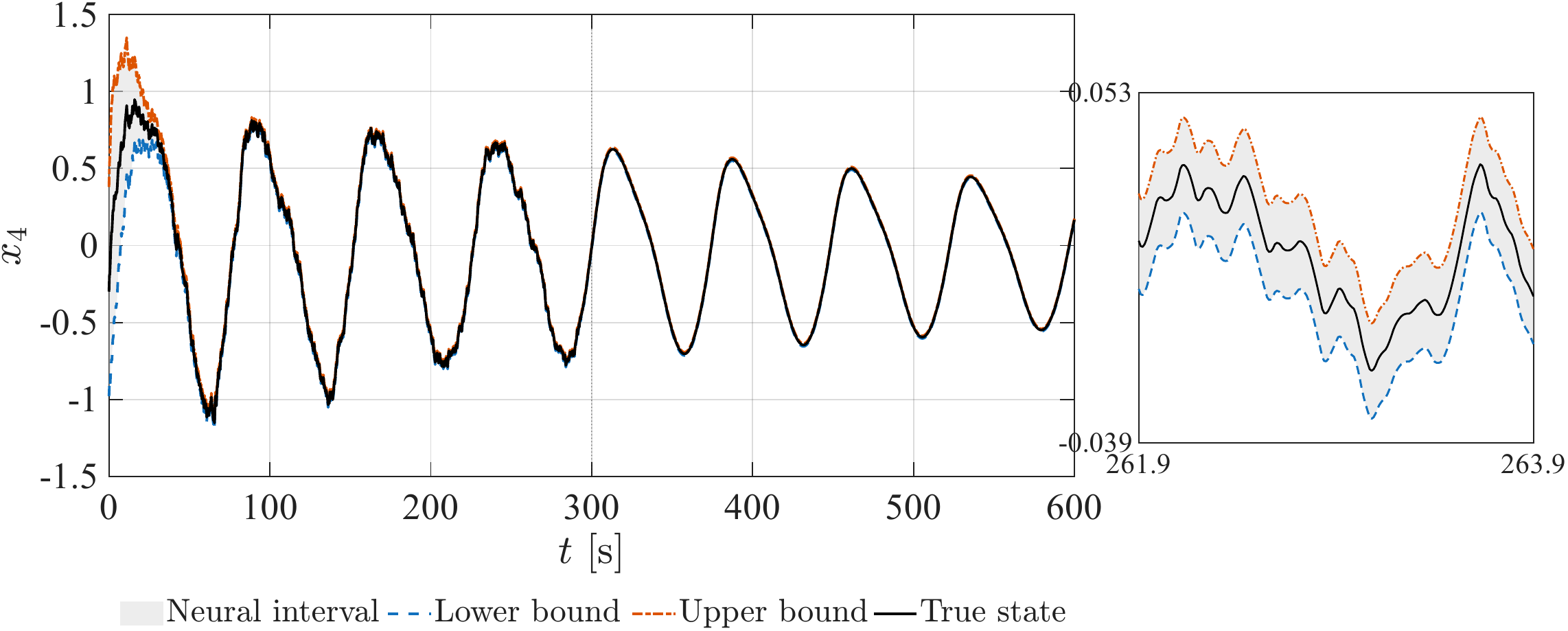}
    \caption{Neural interval estimate for state \(x_4\) at Node~3.}
    \label{fig:neural_interval_x4}
\end{figure}

\FloatBarrier


\begin{figure}[!ht]
    \centering
    \includegraphics[width=\columnwidth]
    {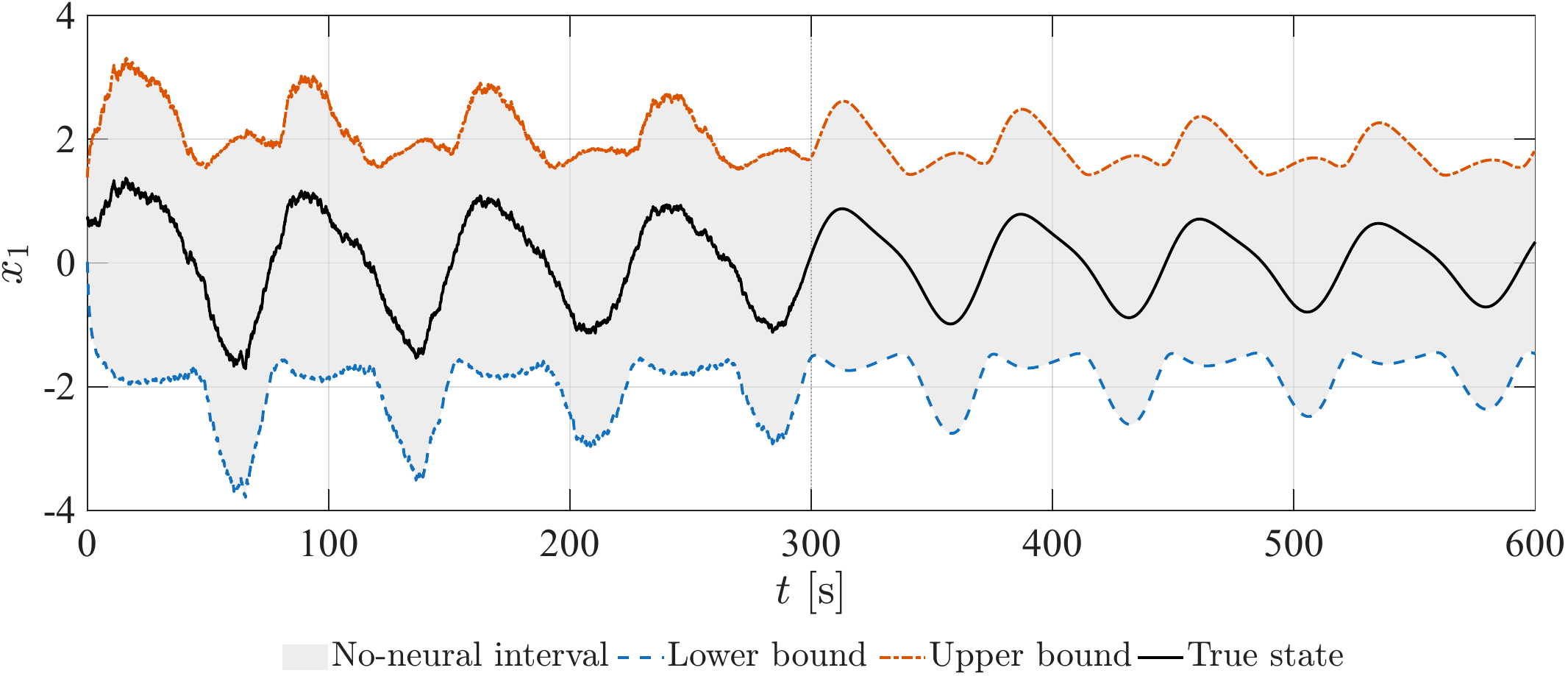}
    \caption{Nominal interval estimate without neural compensation for state \(x_1\) at Node~3.}
    \label{fig:noneural_interval_x1}
\end{figure}

\begin{figure}[!ht]
    \centering
    \includegraphics[width=\columnwidth]
    {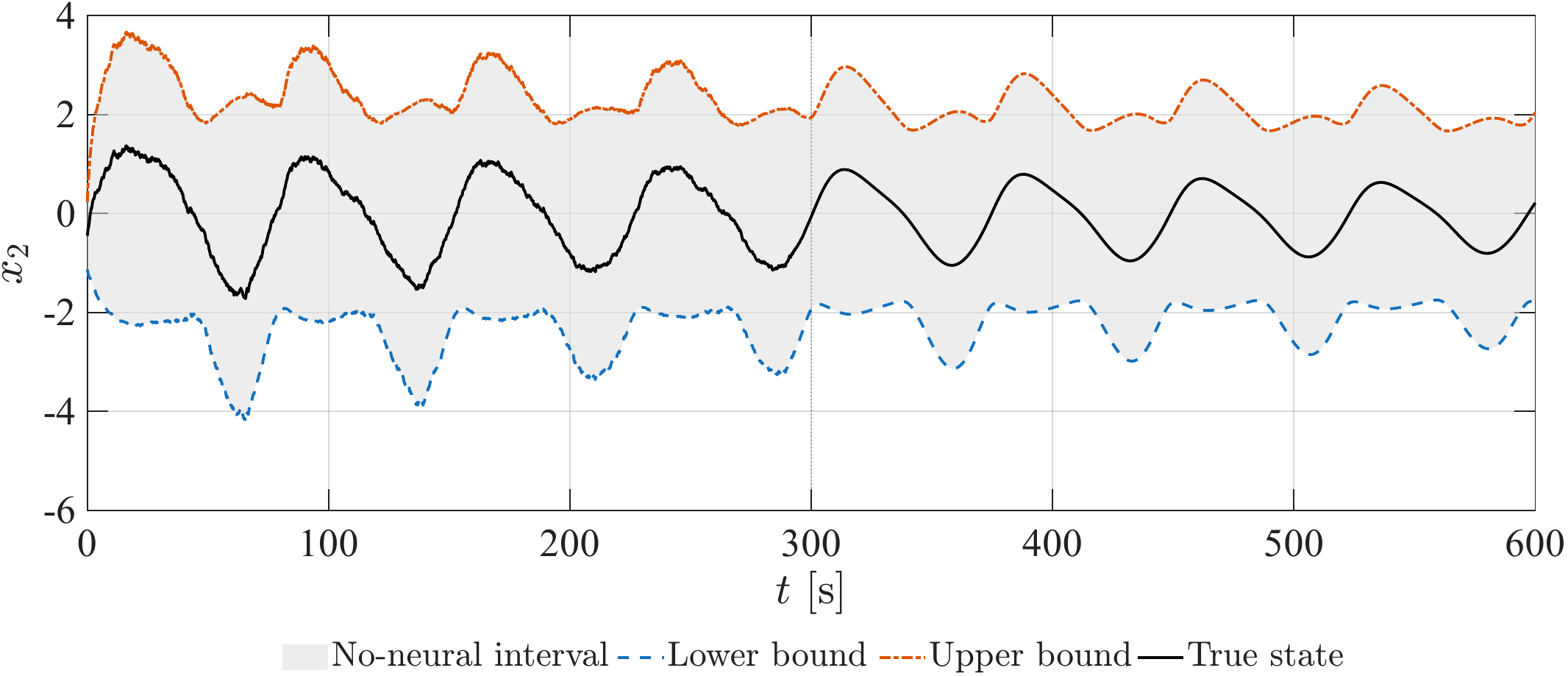}
    \caption{Nominal interval estimate without neural compensation for state \(x_2\) at Node~3.}
    \label{fig:noneural_interval_x2}
\end{figure}

\begin{figure}[!ht]
    \centering
    \includegraphics[width=\columnwidth]
    {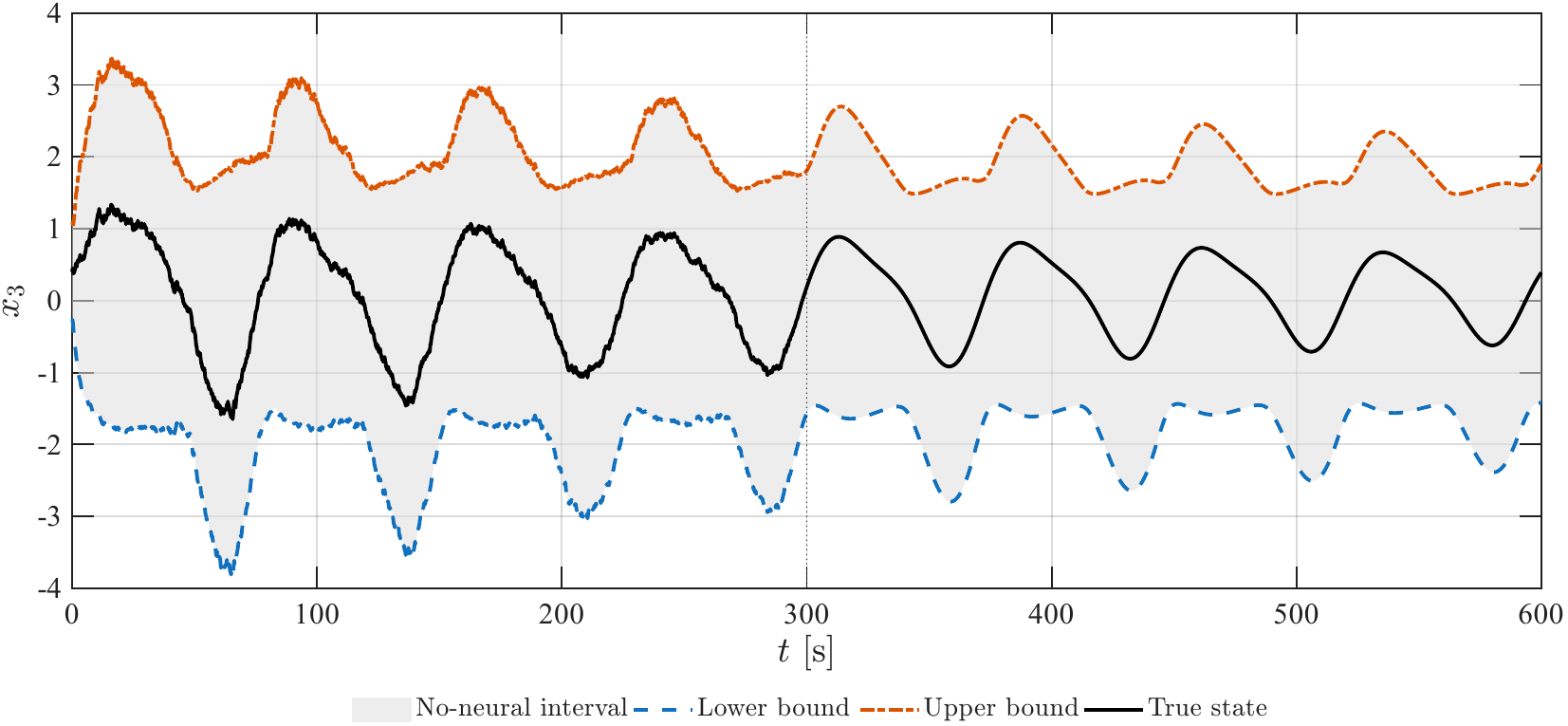}
    \caption{Nominal interval estimate without neural compensation for state \(x_3\) at Node~3.}
    \label{fig:noneural_interval_x3}
\end{figure}

\begin{figure}[!ht]
    \centering
    \includegraphics[width=0.90\columnwidth]
    {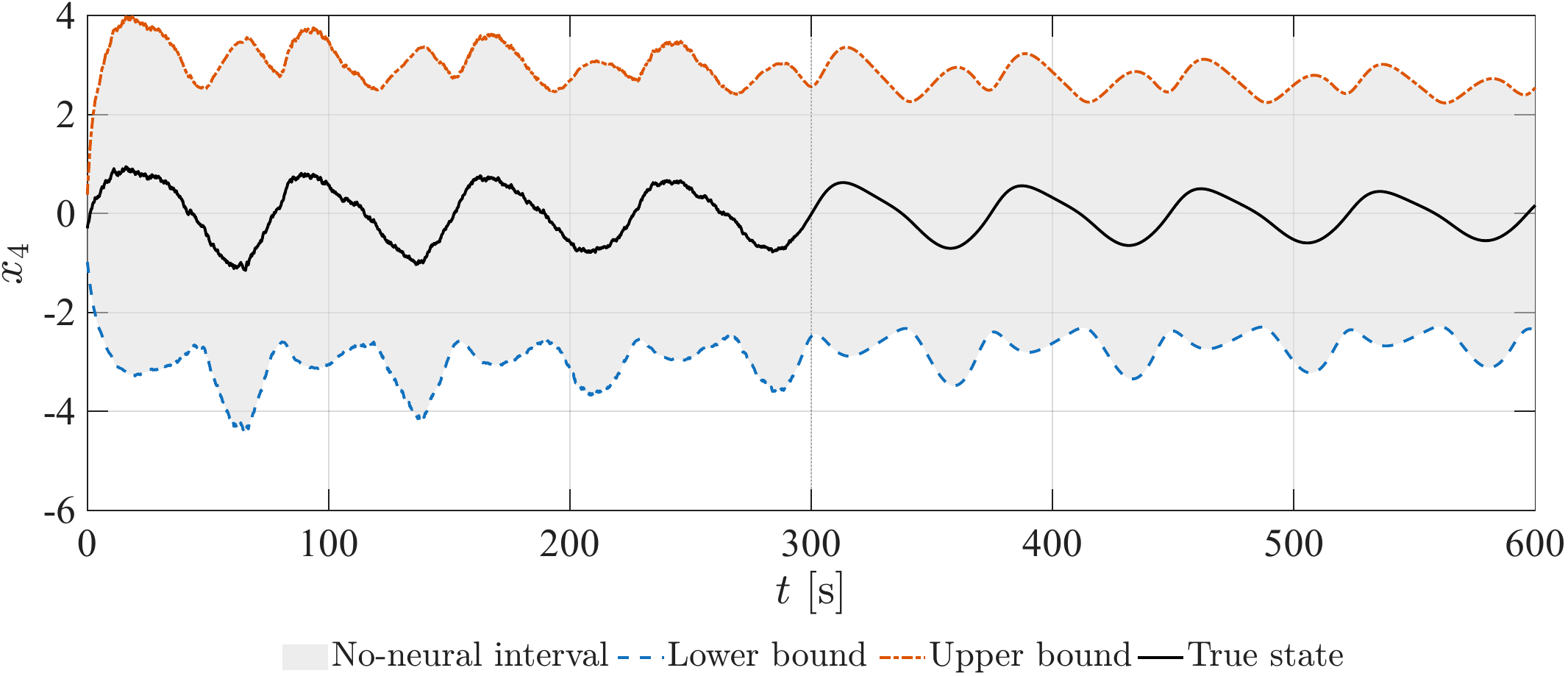}
    \caption{Nominal interval estimate without neural compensation for state \(x_4\) at Node~3.}
    \label{fig:noneural_interval_x4}
\end{figure}

\FloatBarrier


\begin{figure}[!ht]
    \centering
    \includegraphics[width=0.90\columnwidth]
    {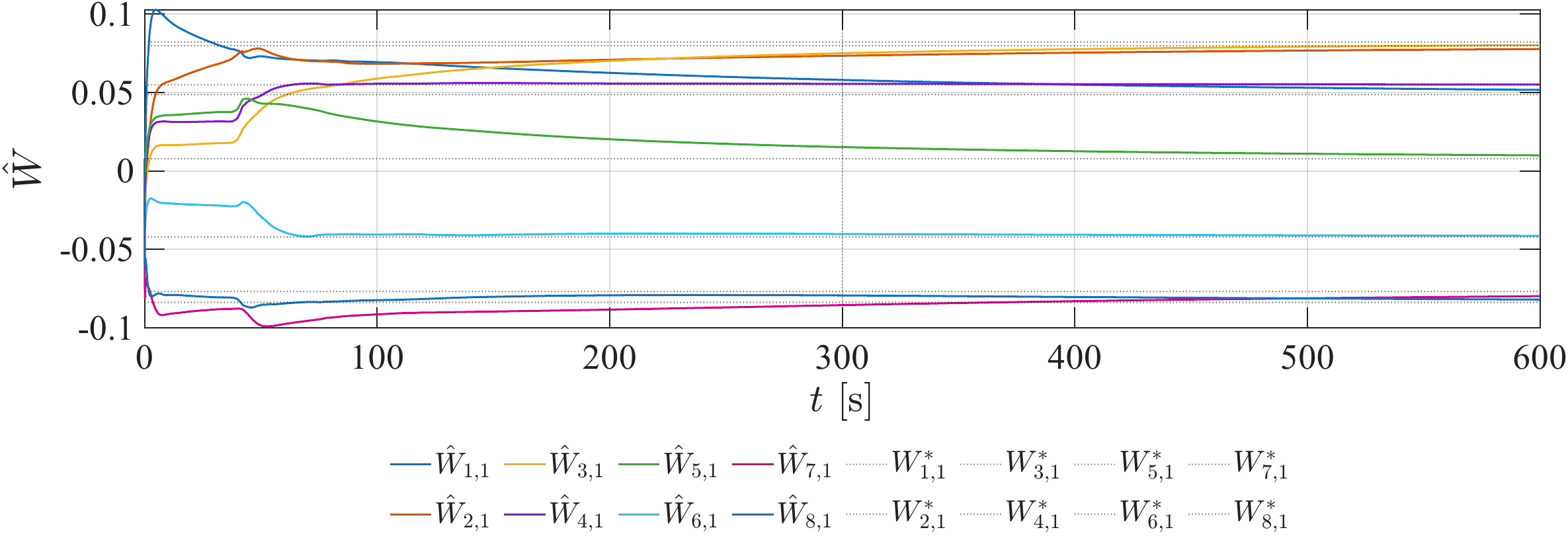}
    \caption{Adaptive neural-weight trajectories associated with \(f_1\) at Node~3.}
    \label{fig:weights_f1_node3}
\end{figure}

\begin{figure}[!ht]
    \centering
    \includegraphics[width=0.90\columnwidth]
    {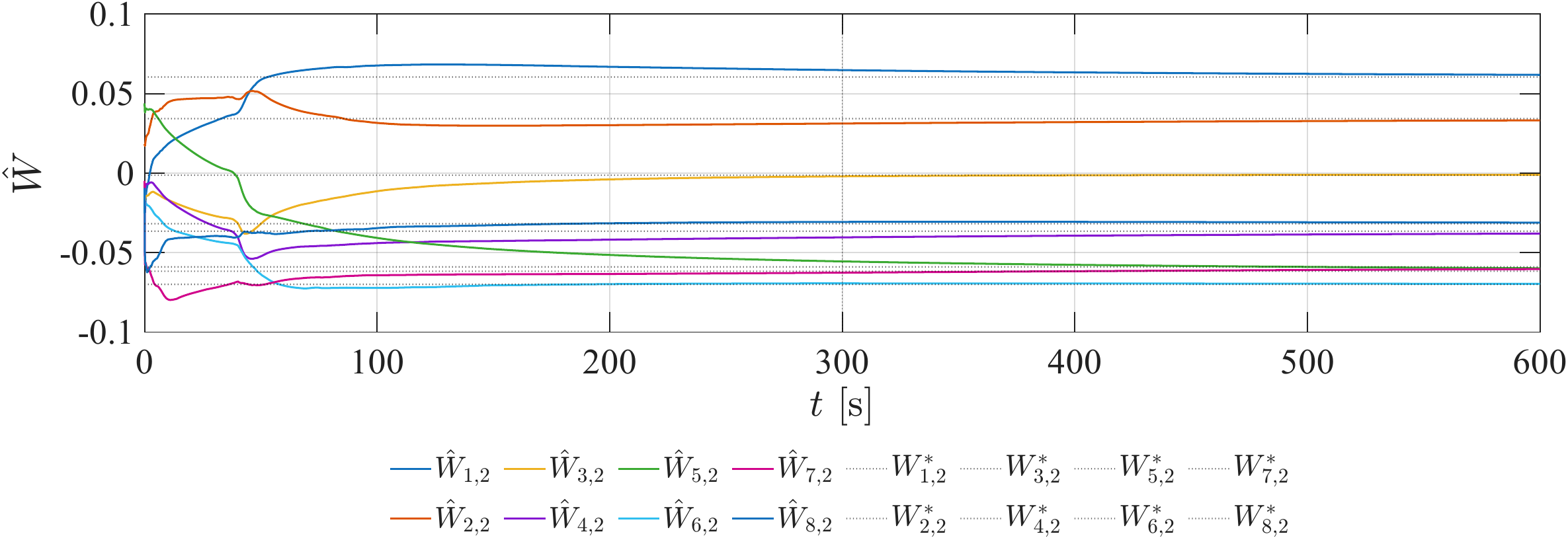}
    \caption{Adaptive neural-weight trajectories associated with \(f_2\) at Node~3.}
    \label{fig:weights_f2_node3}
\end{figure}

\begin{figure}[!ht]
    \centering
    \includegraphics[width=0.90\columnwidth]
    {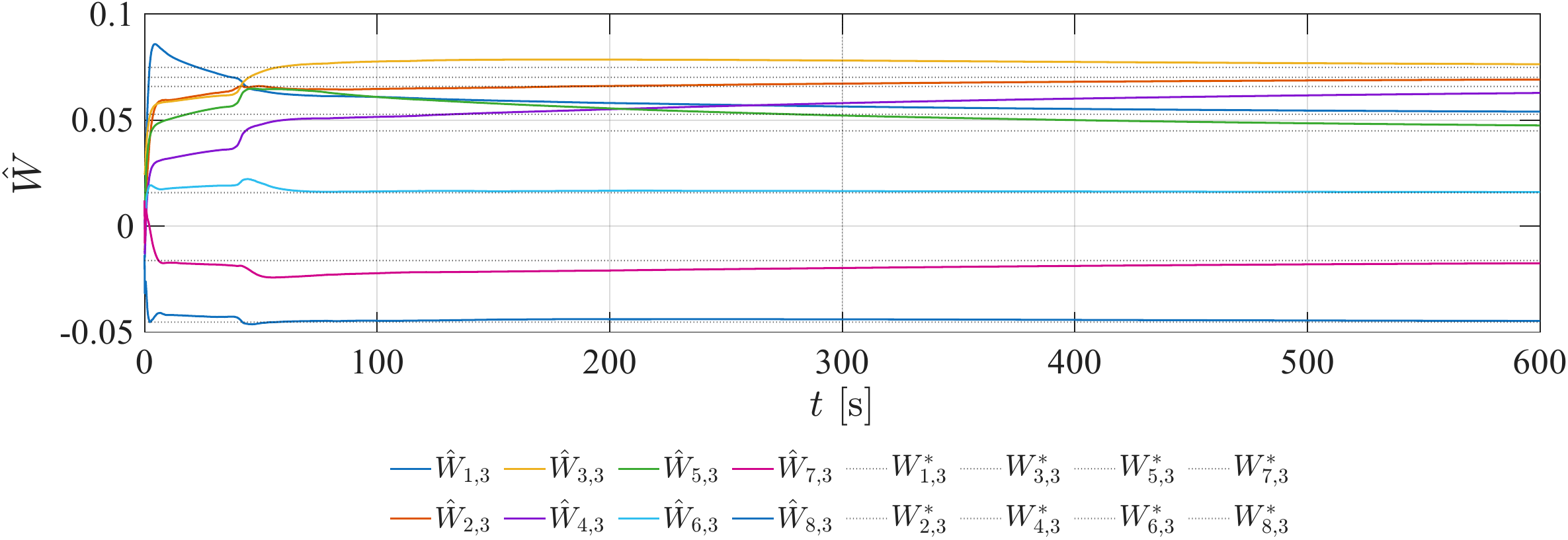}
    \caption{Adaptive neural-weight trajectories associated with \(f_3\) at Node~3.}
    \label{fig:weights_f3_node3}
\end{figure}

\begin{figure}[!ht]
    \centering
    \includegraphics[width=0.90\columnwidth]
    {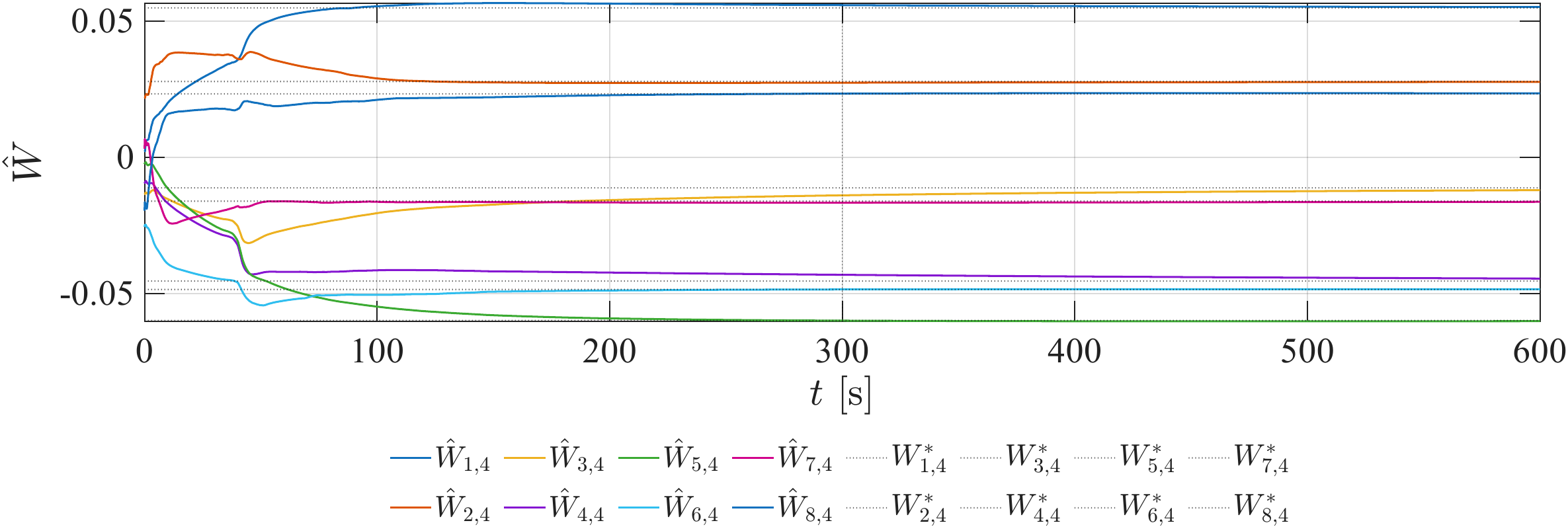}
    \caption{Adaptive neural-weight trajectories associated with \(f_4\) at Node~3.}
    \label{fig:weights_f4_node3}
\end{figure}

\FloatBarrier


\begin{figure}[!ht]
    \centering
    \includegraphics[width=0.90\columnwidth]
    {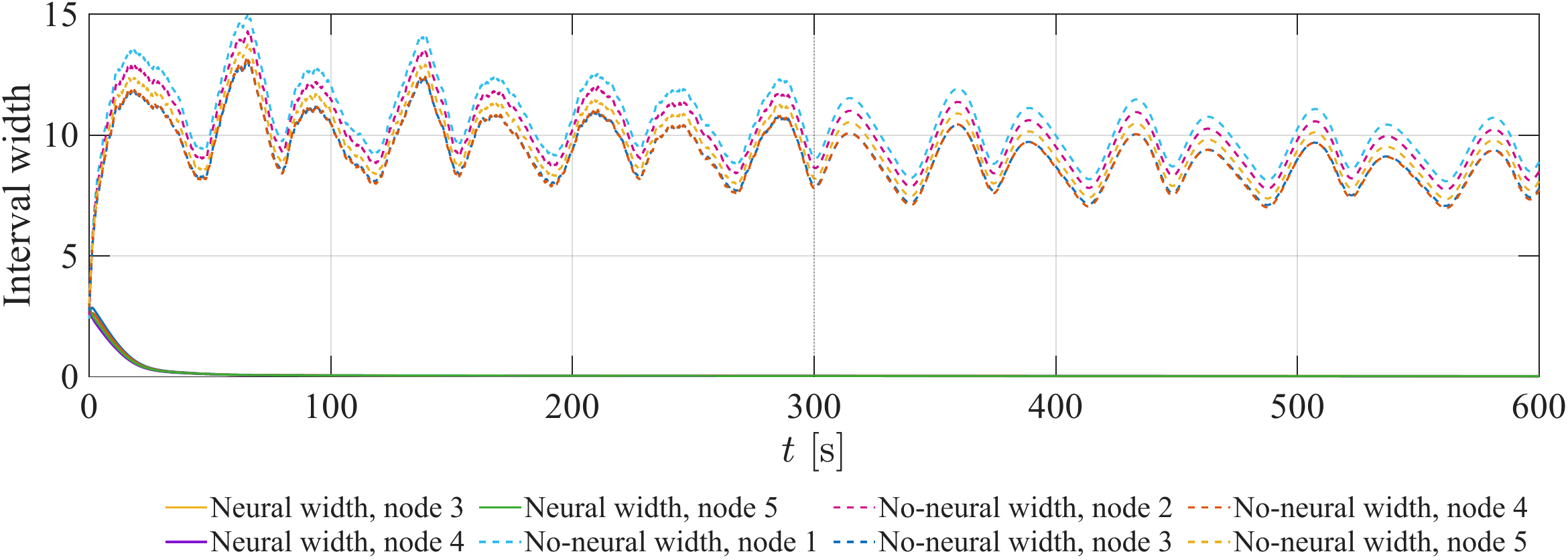}
    \caption{Interval-width comparison for all nodes with and without neural compensation.}
    \label{fig:all_node_width_comparison}
\end{figure}

\FloatBarrier

The neural interval estimates in Figs.~\ref{fig:neural_interval_x1}--\ref{fig:neural_interval_x4} show that the initial enclosures contract rapidly within approximately the first \(30\!-\!50\,\mathrm{s}\) and subsequently remain very tight for all four states. In contrast, the nominal interval observer without neural compensation in Figs.~\ref{fig:noneural_interval_x1}--\ref{fig:noneural_interval_x4} converges only to a considerably wider nonzero tube because the nominal model does not compensate for the unknown nonlinear dynamics. The adaptive-weight trajectories in Figs.~\ref{fig:weights_f1_node3}--\ref{fig:weights_f4_node3} become nearly stationary after approximately \(100\!-\!200\,\mathrm{s}\), while the network-wide comparison in Fig.~\ref{fig:all_node_width_comparison} shows that the neural-compensated interval widths are reduced by more than one order of magnitude relative to the nominal case after the initial transient. These tighter distributed intervals are directly relevant to fault-detection and isolation applications, since they provide less conservative residual thresholds for detecting sensor, actuator, or communication faults in networked systems.

\section{Conclusion}

This paper developed a distributed adaptive neural interval observer for unknown nonlinear systems. The proposed design guarantees bounded transformed and physical-state interval estimation errors together with bounded adaptive neural weights, while the stored-data learning mechanism improves parameter convergence without requiring persistent excitation of the online regressors. The simulation results verify that the neural-compensated intervals contract to significantly tighter bounds than the nominal interval observer and that the adaptive weights approach steady values under finite excitation. These results support the effectiveness of the proposed framework for distributed state estimation and fault-detection applications in networked nonlinear systems.

\section*{Acknowledgment}
The authors would like to thank Ho Chi Minh City University of Technology (HCMUT) and Vietnam National University Ho Chi Minh City (VNU-HCM) for supporting this research.



\begin{thebibliography}{99}

\bibitem{Gouze2000}
J.-L. Gouz\'e, A. Rapaport, and M.~Z. Hadj-Sadok,
``Interval observers for uncertain biological systems,''
\emph{Ecological Modelling},
vol.~133, nos.~1--2, pp.~45--56, 2000,
doi: 10.1016/S0304-3800(00)00279-9.

\bibitem{Efimov2013}
D. Efimov, T. Ra\"issi, S. Chebotarev, and A. Zolghadri,
``Interval state observer for nonlinear time varying systems,''
\emph{Automatica},
vol.~49, no.~1, pp.~200--205, 2013,
doi: 10.1016/j.automatica.2012.07.004.

\bibitem{ZhangShen2023}
Z. Zhang and J. Shen,
``A survey on interval observer design using positive system approach,''
\emph{Franklin Open},
vol.~4, Art.~no.~100031, 2023,
doi: 10.1016/j.fraope.2023.100031.

\bibitem{Mitra2018}
A. Mitra and S. Sundaram,
``Distributed observers for LTI systems,''
\emph{IEEE Trans. Autom. Control},
vol.~63, no.~11, pp.~3689--3704, Nov.~2018,
doi: 10.1109/TAC.2018.2798998.

\bibitem{Han2019}
W. Han, H.~L. Trentelman, Z. Wang, and Y. Shen,
``A simple approach to distributed observer design for linear systems,''
\emph{IEEE Trans. Autom. Control},
vol.~64, no.~1, pp.~329--336, Jan.~2019,
doi: 10.1109/TAC.2018.2828103.

\bibitem{Wang2025Distributed}
X. Wang, W. Xu, H. Su, Z. Gao, and G. Chen,
``Designing a completely distributed interval observer for the LTI
system,''
\emph{IEEE Trans. Autom. Control},
vol.~70, no.~3, pp.~1793--1808, Mar.~2025,
doi: 10.1109/TAC.2024.3471332.

\bibitem{Wang2026Descriptor}
X. Wang, F. Li, and Y. Zhou,
``Distributed interval observer design for descriptor systems,''
\emph{J. Franklin Inst.},
vol.~363, no.~8, Art.~no.~108677, 2026,
doi: 10.1016/j.jfranklin.2026.108677.

\bibitem{Yi2022}
Z. Yi, W. Xie, L. Liu, and B. Xu,
``A neural network adaptive interval observer design for nonlinear
systems,''
\emph{IET Control Theory Appl.},
vol.~16, no.~6, pp.~615--624, 2022,
doi: 10.1049/cth2.12258.



\bibitem{Chowdhary2010}
G. Chowdhary and E. N. Johnson,
``Concurrent learning for convergence in adaptive control without persistency of excitation,''
in \emph{Proc. 49th IEEE Conf. Decision and Control},
2010, pp. 3674--3679.

\bibitem{Parikh2019}
A. Parikh, R. Kamalapurkar, and W. E. Dixon,
``Integral concurrent learning: Adaptive control with parameter convergence using finite excitation,''
\emph{International Journal of Adaptive Control and Signal Processing},
vol. 33, no. 12, pp. 1775--1787, 2019.


\end{thebibliography}
\end{document}